\documentclass[11pt]{article}
\usepackage[margin=0.75in]{geometry}
\pdfoutput=1
\usepackage[utf8]{inputenc}
\usepackage[english]{babel}
\usepackage[T1]{fontenc}
\usepackage{amsmath,amssymb}
\usepackage{graphicx}           
\usepackage{hyperref}
\usepackage{relsize}
\usepackage{thm-restate}
\hypersetup{
    colorlinks=true,
    linkcolor=blue,
    citecolor=magenta,
    filecolor=magenta,      
    urlcolor=blue,
    pdftitle={},
    pdfpagemode=FullScreen,
    }
\usepackage{soul}
\usepackage{mathtools}

\usepackage{tikz}
\usetikzlibrary{shapes}
\usepackage{xcolor}
\usepackage{amsthm}
\theoremstyle{definition}
\usepackage{float}
\makeatletter
\newlength\min@xx
\newcommand*\xxrightarrow[1]{\begingroup
  \settowidth\min@xx{$\m@th\scriptstyle#1$}
  \@xxrightarrow}
\newcommand*\@xxrightarrow[2][]{
  \sbox8{$\m@th\scriptstyle#1$}  
  \ifdim\wd8>\min@xx \min@xx=\wd8 \fi
  \sbox8{$\m@th\scriptstyle#2$} 
  \ifdim\wd8>\min@xx \min@xx=\wd8 \fi
  \xrightarrow[{\mathmakebox[\min@xx]{\scriptstyle#1}}]
    {\mathmakebox[\min@xx]{\scriptstyle#2}}
  \endgroup}
\makeatother

\newcommand{\unitx}{\mathbf{e}_x}
\newcommand{\unity}{\mathbf{e}_y}

\newcommand{\consistent}{\stackrel{c}{=}}
\newcommand{\maxmerge}{\Join}
\newcommand{\rightmerge}{\lhd}
\newcommand{\mergeprod}[2]{\underset{\scriptsize #1}{\stackrel{\scriptsize #2}{\text{\scalebox{1.75}{$\rightmerge$}}}}}

\newcommand{\mono}[0]{\xxrightarrow{MMM}{\text{Mono}}}
\newcommand{\revmono}[0]{\xxrightarrow{MMM}{\text{Mono}^{-1}}}
\newcommand{\petz}[0]{\xxrightarrow{MMM}{\text{Petz}}}

\newcommand{\entropy}[1]{S\left(#1\right)}
\newcommand{\markoventropy}[2]{S_{M}^{#2}\left(#1\right)}

\newcommand{\centeredTikZ}[1]{
\begin{tikzpicture}[scale=0.55, baseline={([yshift=-.55ex]current bounding box.center)}]
#1
\end{tikzpicture}
}

\newcommand{\emptysquare}[2]{\draw[fill=white] (#1, #2) rectangle (#1 +0.5, #2 + 0.5)}

\newcommand{\squaresempty}[4]{
\foreach \x in {1,...,#1}
{
\foreach \y in {1,...,#2}
{
\emptysquare{#3*0.5 - #4*0.25 +0.5*\x - 0.25*\y}{#4*0.5+0.5*\y};
}
}
}

\newcommand{\shearedTikZ}[1]
           {
             \begin{tikzpicture}[x={(1,0)}, y={(-0.5,0.866)}, scale=0.5, baseline={([yshift=0.1ex]current bounding box.center)}]
               #1
             \end{tikzpicture}
           }
\newcommand{\shearedTikZmini}[1]
           {
             \begin{tikzpicture}[x={(1,0)}, y={(-0.5,0.866)}, scale=0.3, baseline={([yshift=0.1ex]current bounding box.center)}]
               #1
             \end{tikzpicture}
           }
\newcommand{\mygrid}[1]
           {
             \foreach \c in {0,1,2}
                      {
                        \draw[gridline] (\c,0) -- (\c,2);
                        \draw[gridline] (0,\c) -- (2,\c);
                      }
              \node[below] at (2,-0.1) () {\scriptsize $#1$};
           }
\newcommand{\mycustomgrid}[3]
           {
             \foreach \x in {1, ..., #1}
                      {
                        \draw[gridline] (\x,1) -- (\x,#2);
                      }
            \foreach \y in {1, ..., #2}
            {
            \draw[gridline] (1, \y) -- (#1, \y);
            }
              \node[below] at (#1,1) () {\scriptsize $#3$};
           }

\newcommand{\nbym}[3]
           {
             \shearedTikZ{
               \foreach \x in {1, ..., #1}
                        {
                          \foreach \y in {1, ..., #2}
                                   {
                                     \node[regular] at (\x, \y) () {};
                                   }
                        }
                        \foreach \x in {1, ..., #1}
                                 {
                                   \draw[gridline] (\x, 1) -- (\x, #2);
                                 }
                                 \foreach \y in {1, ..., #2}
                                          {
                                            \draw[gridline] (1, \y) -- (#1, \y);
                                          }
                                          \node[below, text centered] at (#1,1) () {\scriptsize $#3$};
             }}
           
\newcommand{\snakeonehooked}[2]{
  \shearedTikZ{
  \node[regular] at (1, 1) (hook) {};
    \node[regular] at (1,0) (leftmiddle) {};
    \node[] at (1.75, 0) (leftin) {};
    \node[] at (2.25, 0) (rightin) {};
    \node[regular] at (3, 0) (rightmiddle) {};
    \draw[gridline] (leftmiddle) -- (leftin);
    \draw[gridline] (rightin) -- (rightmiddle);
    \draw[dotted] (leftmiddle) -- (rightmiddle);
    \draw[gridline] (leftmiddle) -- (hook);
    \node[below] at (1.0, 0) () {\scriptsize $#1\vphantom{-\unitx}$};
    \node[below] at (3, 0) () {\scriptsize $#2\vphantom{-\unitx}$};
    }}
           
\newcommand{\snakeone}[2]{
  \shearedTikZ{
    \node[regular] at (1,0) (leftmiddle) {};
    \node[] at (1.75, 0) (leftin) {};
    \node[] at (2.25, 0) (rightin) {};
    \node[regular] at (3, 0) (rightmiddle) {};
    \draw[gridline] (leftmiddle) -- (leftin);
    \draw[gridline] (rightin) -- (rightmiddle);
    \draw[dotted] (leftmiddle) -- (rightmiddle);
    \node[below] at (1.0, 0) () {\scriptsize $#1\vphantom{-\unitx}$};
    \node[below] at (3, 0) () {\scriptsize $#2\vphantom{-\unitx}$};
    }}

\newcommand{\snaketwo}[2]{
  \shearedTikZ{
    \begin{scope}[xshift=-0.5cm, yshift=0.866cm]
    \node[regular] at (1,0) (leftmiddleu) {};
    \node[] at (1.75, 0) (leftinu) {};
    \node[] at (2.25, 0) (rightinu) {};
    \node[regular] at (3, 0) (rightmiddleu) {};
    \draw[gridline] (leftmiddleu) -- (leftinu);
    \draw[gridline] (rightinu) -- (rightmiddleu);
    \end{scope}
    \node[regular] at (1,0) (leftmiddle) {};
    \node[] at (1.75, 0) (leftin) {};
    \node[] at (2.25, 0) (rightin) {};
    \node[regular] at (3, 0) (rightmiddle) {};
    \draw[gridline] (leftmiddle) -- (leftin);
    \draw[gridline] (rightin) -- (rightmiddle);
    \node[below] at (1.0, 0) () {\scriptsize $#1\vphantom{-\unitx}$};
    \node[below] at (3.0, 0) () {\scriptsize $#2\vphantom{-\unitx}$};
    \draw[gridline] (leftmiddleu) -- (leftmiddle);
    \draw[gridline] (rightmiddleu) -- (rightmiddle);
    \draw[dotted] (leftmiddle) -- (rightmiddle);
    \draw[dotted] (leftmiddleu) -- (rightmiddleu);
    }}

\newcommand{\snaketwohooked}[2]{
  \shearedTikZ{
    \node[regular] at (1,2) (hook) {};
    \begin{scope}[xshift=-0.5cm, yshift=0.866cm]
    \node[regular] at (1,0) (leftmiddleu) {};
    \node[] at (1.75, 0) (leftinu) {};
    \node[] at (2.25, 0) (rightinu) {};
    \node[regular] at (3, 0) (rightmiddleu) {};
    \draw[gridline] (leftmiddleu) -- (leftinu);
    \draw[gridline] (rightinu) -- (rightmiddleu);
    \end{scope}
    \node[regular] at (1,0) (leftmiddle) {};
    \node[] at (1.75, 0) (leftin) {};
    \node[] at (2.25, 0) (rightin) {};
    \node[regular] at (3, 0) (rightmiddle) {};
    \draw[gridline] (leftmiddle) -- (leftin);
    \draw[gridline] (rightin) -- (rightmiddle);
    \node[below] at (1.0, 0) () {\scriptsize $#1\vphantom{-\unitx}$};
    \node[below] at (3.0, 0) () {\scriptsize $#2\vphantom{-\unitx}$};
    \draw[gridline] (leftmiddleu) -- (leftmiddle);
    \draw[gridline] (rightmiddleu) -- (rightmiddle);
    \draw[dotted] (leftmiddle) -- (rightmiddle);
    \draw[dotted] (leftmiddleu) -- (rightmiddleu);
    \draw[gridline] (leftmiddleu) -- (hook);
    }}
\newcommand{\snakethree}[2]{
  \shearedTikZ{
  \begin{scope}[xshift=-1cm, yshift=1.732cm]
    \node[regular] at (1,0) (leftmiddletop) {};
    \node[] at (1.75, 0) (leftintop) {};
    \node[] at (2.25, 0) (rightintop) {};
    \node[regular] at (3, 0) (rightmiddletop) {};
    \draw[gridline] (leftmiddletop) -- (leftintop);
    \draw[gridline] (rightintop) -- (rightmiddletop);
    \end{scope}
    \begin{scope}[xshift=-0.5cm, yshift=0.866cm]
    \node[regular] at (1,0) (leftmiddleu) {};
    \node[] at (1.75, 0) (leftinu) {};
    \node[] at (2.25, 0) (rightinu) {};
    \node[regular] at (3, 0) (rightmiddleu) {};
    \draw[gridline] (leftmiddleu) -- (leftinu);
    \draw[gridline] (rightinu) -- (rightmiddleu);
    \end{scope}
    \node[regular] at (1,0) (leftmiddle) {};
    \node[] at (1.75, 0) (leftin) {};
    \node[] at (2.25, 0) (rightin) {};
    \node[regular] at (3, 0) (rightmiddle) {};
    \draw[gridline] (leftmiddle) -- (leftin);
    \draw[gridline] (rightin) -- (rightmiddle);
    \node[below] at (1.0, 0) () {\scriptsize $#1\vphantom{-\unitx}$};
    \node[below] at (3.0, 0) () {\scriptsize $#2\vphantom{-\unitx}$};
    \draw[gridline] (leftmiddletop)-- (leftmiddleu) -- (leftmiddle);
    \draw[gridline] (rightmiddletop)--(rightmiddleu) -- (rightmiddle);
    \draw[dotted] (leftmiddle) -- (rightmiddle);
    \draw[dotted] (leftmiddleu) -- (rightmiddleu);
    \draw[dotted] (leftmiddletop) -- (rightmiddletop);
    }}
\newcommand{\snakethreecmi}[2]{
  \shearedTikZ{
  \begin{scope}[xshift=-1cm, yshift=1.732cm]
    \node[triangle] at (1,0) (leftmiddletop) {};
    \node[] at (1.75, 0) (leftintop) {};
    \node[] at (2.25, 0) (rightintop) {};
    \node[triangle] at (3, 0) (rightmiddletop) {};
    \draw[gridline] (leftmiddletop) -- (leftintop);
    \draw[gridline] (rightintop) -- (rightmiddletop);
    \end{scope}
    \begin{scope}[xshift=-0.5cm, yshift=0.866cm]
    \node[square] at (1,0) (leftmiddleu) {};
    \node[] at (1.75, 0) (leftinu) {};
    \node[] at (2.25, 0) (rightinu) {};
    \node[square] at (3, 0) (rightmiddleu) {};
    \draw[gridline] (leftmiddleu) -- (leftinu);
    \draw[gridline] (rightinu) -- (rightmiddleu);
    \end{scope}
    \node[disk] at (1,0) (leftmiddle) {};
    \node[] at (1.75, 0) (leftin) {};
    \node[] at (2.25, 0) (rightin) {};
    \node[disk] at (3, 0) (rightmiddle) {};
    \draw[gridline] (leftmiddle) -- (leftin);
    \draw[gridline] (rightin) -- (rightmiddle);
    \node[below] at (1.0, 0) () {\scriptsize $#1\vphantom{-\unitx}$};
    \node[below] at (3.0, 0) () {\scriptsize $#2\vphantom{-\unitx}$};
    \draw[gridline] (leftmiddletop)-- (leftmiddleu) -- (leftmiddle);
    \draw[gridline] (rightmiddletop)--(rightmiddleu) -- (rightmiddle);
    \draw[dotted] (leftmiddle) -- (rightmiddle);
    \draw[dotted] (leftmiddleu) -- (rightmiddleu);
    \draw[dotted] (leftmiddletop) -- (rightmiddletop);
    }}

\tikzset{
  triangle/.style = {regular polygon, regular polygon sides=3,
              draw, fill=red!30!white, text width=1em,
              inner sep=0mm, outer sep=0mm, scale=0.3},
  square/.style = {regular polygon, regular polygon sides=4,
              draw, fill=green!30!white, text width=1em,
              inner sep=0mm, outer sep=0mm, scale=0.4},
  disk/.style = {circle, draw, fill=blue!30!white, text width=1em,
    inner sep=0mm, outer sep=0mm, scale=0.6},
  regular/.style = {circle, draw, fill=black, text width=1em,
    inner sep=0mm, outer sep=0mm, scale=0.4},
  gridline/.style = {line width=0.5pt}
  }

\usepackage{lipsum}

\newtheorem{theorem}{Theorem}
\newtheorem{lemma}{Lemma}
\newtheorem{corollary}{Corollary}
\newtheorem{definition}{Definition}

\begin{document}

\title{Entropy scaling law and the quantum marginal problem: simplification and generalization}

\author{Isaac H. Kim \thanks{Department of Computer Science, UC Davis, Davis, CA 95616, USA}}

\maketitle

\begin{abstract}
Recently, we introduced a solution to the quantum marginal problem relevant to two-dimensional quantum many-body systems~[\href{https://journals.aps.org/prx/abstract/10.1103/PhysRevX.11.021039}{I. H. Kim, Phys. Rev. X, \textbf{11}, 021039}]. 
One of the conditions was that the marginals are internally translationally invariant. We show that this condition can be replaced by a weaker condition, namely the local consistency of the marginals. This extends the applicability of the solution to any quantum many-body states in two dimensions that satisfy the entropy scaling law, with or without symmetry. We also significantly simplify the proof by advocating the usage of the maximum-entropy principle.
\end{abstract}

\section{Introduction}
\label{sec:introduction}
One of the major challenges in modern science lies in the inherent exponential complexity of describing the many-body wave functions that appear in nature. Naively, the number of parameters that need to be specified to even write down the wave function grows exponentially with the number of particles. Overcoming this challenge is one of the outstanding goals of physics, chemistry, and computer science. 

While such an exponential complexity is unavoidable for generic many-body wave functions, there may be better ways to characterize ``physical'' wave functions. For instance, physical quantities of interest, such as the energy and various order parameters, can be calculated directly from the \emph{marginal density matrices} (\emph{marginals} for short), each acting on a finite-dimensional Hilbert space. Characterizing the wave function in terms of the marginals is the primary goal of the \emph{quantum marginal problem}; see Ref.~\cite{Schilling2015} for a review. In the quantum marginal problem, one is given a set of marginals, each defined over a different subset of particles. The main question is whether there exists a global state whose reduced density matrices are equal to the given set of marginals. If the existence of such a state can be guaranteed, we say that the marginals are \emph{consistent}.

Despite many attempts, progress in this direction has been slow. This is likely because the problem of deciding whether a set of marginals is consistent or not is QMA-hard~\cite{Liu2006,Broadbent2019}. Without imposing an extra structure on the problem, it is unlikely that one can find an efficient solution.

However, recently a new nontrivial solution to the quantum marginal problem was found~\cite{Kim2021}. The central overarching theme behind this solution is the \emph{entropy scaling law}. This is a hypothesis that the von Neumann entropy --- defined as $S(\rho)= -\text{Tr}(\rho \log \rho)$ --- for a marginal density matrix $\rho$, is a sum of terms,  each proportional to the volume, area, and a universal constant, respectively. This scaling law, advocated in Ref.~\cite{Kitaev2006,Levin2006} and verified in many examples after that, can lead to surprising conclusions. If the scaling law is satisfied locally on the marginals and if the marginals obey an extra constraint related to the translation symmetry, then one can ensure that the marginals are consistent. Each such constraint is defined on a finite-dimensional Hilbert space, and the number of constraints scales linearly with the system size, making this an ``efficient'' solution.\footnote{We note that the conditions are \emph{equality constraints}, which may be difficult to enforce with a finite precision in realistic computers. Relaxing this to an inequality constraint may be possible using the approach developed in Ref.~\cite{Kim2016} but we do not discuss this issue in this paper.}

This solution has several appealing features. First and foremost, it was demonstrated in Ref.~\cite{Kim2021} that the solution applies to a large class of physical states --- even those with long-range entanglement, \emph{e.g.,} the toric code~\cite{Kitaev2003}. Secondly, because a set of marginals defines the global state, local expectation values can be computed readily and efficiently (from the marginals). That this is possible without sacrificing an ability to accommodate long-range entanglement is a remarkable fact. Lastly, one can compute the maximum global entropy exactly. This is a feature that allows nontrivial calculation for the free energy (at finite temperature), not just for ground state energy.
Thus, the solution in Ref.~\cite{Kim2021} may be a promising tool to study quantum many-body systems at both zero and finite temperatures.

However, that solution was only applicable to translationally invariant states, and removing this requirement was left as an open problem in Ref.~\cite{Kim2021}. We solve this problem in this paper, thereby generalizing the solution substantially. Practically, this means that a larger class of many-body quantum states have an efficient classical description, in the sense that their physical properties such as energy (with respect to a local Hamiltonian) and entropy can be computed efficiently. Our work also provides a new conceptual lesson by clarifying the nature of these solutions. The solution in Ref.~\cite{Kim2021} required a symmetry \emph{and} entropy scaling law. As such, it is natural to ask which of the two plays a more important role. Since we managed to remove all the symmetries, we can now definitively say that the entropy scaling law is the important one, not the symmetry.


The rest of this paper is structured as follows. In Section~\ref{sec:summary}, we provide an executive summary of our conventions and the main results.  Section~\ref{sec:conditional_independence} introduces the requisite background in conditional independence, the main workhorse of our theory. In Section~\ref{sec:recursion} and~\ref{sec:final}, we prove the main results. We end with a brief discussion in Section~\ref{sec:discussion}.

\section{Summary}
\label{sec:summary}
Our main result is a nontrivial solution to the quantum marginal problem for interacting quantum many-body systems in two dimensions. Suppose we are given a set of marginals over balls of bounded radii. We formulate a set of conditions that ensures the existence of a global state consistent with those marginals. The number of such conditions scales linearly with the number of degrees of freedom being considered, and each condition can be verified in a constant time. Thus, the conditions as a whole can be verified in a time that scales linearly with the number of elementary degrees of freedom. In the rest of this section, we provide a more detailed explanation of this result. Our notations and conventions are introduced in Section~\ref{sec:notation}. Our main results are summarized in Section~\ref{sec:results}.

\subsection{Notations and conventions}
\label{sec:notation}

Consider a Hilbert space $\mathcal{H}$ with a tensor product structure over a set $\Lambda$, \emph{i.e.,} $\mathcal{H} = \bigotimes_{v\in \Lambda} \mathcal{H}_v$. For $A\subset \Lambda$, let $\mathcal{D}_A$ be the set of density matrices acting on $\mathcal{H}_A := \otimes_{v\in A} \mathcal{H}_v$. We say $\rho_A \in \mathcal{D}_A$ and $\sigma_B \in \mathcal{D}_B$ are \emph{locally consistent} if $\rho_{A\cap B} = \sigma_{A\cap B}$ and denote this fact as
\begin{equation}
    \rho_A \consistent \sigma_B.
\end{equation}

For a density matrix $\rho$, its von Neumann entropy is defined as $S(\rho) := -\text{Tr}(\rho \log \rho)$. For a bipartite density matrix $\rho_{AB} \in \mathcal{D}_{AB}$, its conditional entropy is defined as $S(A|B)_{\rho} = S(\rho_{AB}) - S(\rho_B)$. 
For a tripartite density matrix $\rho_{ABC} \in \mathcal{D}_{ABC}$, its conditional mutual information is
$I(A:C|B)_{\rho} := S(A|B)_{\rho} - S(A|BC)_{\rho}$. By the strong subadditivity of entropy (SSA), $I(A:C|B)_{\rho}\geq 0$ for any density matrix $\rho$~\cite{Lieb1973}. If SSA is satisfied with an equality, the underlying state is said to be a \emph{quantum Markov chain}, specifically, between $A$ and $C$ conditioned on $B$.

Now let us discuss our choice of $\Lambda$. For concreteness, we shall set $\Lambda = \{v_x \unitx + v_y \unity: v_x, v_y \in \mathbb{Z} \},$ where $\unitx = \hat{x}$ and $\unity = -\frac{1}{2}\hat{x} + \frac{\sqrt{3}}{2}\hat{y}$, with $\hat{x}$ and $\hat{y}$ being the unit vectors in $\mathbb{R}^2$. We shall denote the element $v_x \unitx + v_y \unity \in \Lambda$ as  $v=(v_x, v_y)$. We can interpret this lattice as a planar graph, viewing $\Lambda$ as a set of vertices and also by assigning edges between $v$ and $u$ such that $|v_x-u_x| + |v_y-u_y| =1$. 

For $A \subset \Lambda$, its \emph{anchoring point} is, colloquially speaking, the right-most vertex on the bottom row of $A$. (More formally, this is the $v\in A$ with the largest $v_x$ within the set of vertices in $A$ with the smallest $v_y$.) If $v$ is an anchoring point of $A$, we say $A$ is \emph{anchored at $v$}. A $n\times m$ \emph{cluster} is a subset of $\Lambda$ consisting of points $\{(v_x + k_x, v_y+k_y): k_x \in \{ 1, \ldots, n\}, k_y \in \{1, \ldots m \} \}$ for some $v \in \Lambda$. Naturally, the anchoring point of this particular cluster would be $(v_x+n, v_y)$.

We shall assume that we are given a set of density matrices on every $3\times 3$ cluster and refer to these marginals as  \emph{fundamental marginals}. (We shall denote the set of fundamental marginals as $\mathcal{M}$.) We shall represent each  fundamental marginal by the following diagram:
\begin{equation}
    \nbym{3}{3}{v}, \label{eq:fundamental_marginal_ex1}
\end{equation}
where $v\in \Lambda$ is the anchoring point of the underlying cluster. (If the diagram appears within the text, we shall simply state the anchoring point without specifying it explicitly in the diagram, \emph{e.g.,} $\shearedTikZmini{\mycustomgrid{2}{2}{} \node[regular] at (1,1) () {}; \node[regular] at (2,1) () {}; \node[regular] at (1,2) () {}; \node[regular] at (2,2) () {};}$ anchored at $v$.) Reduced density matrices of the fundamental marginals shall be also referred to as fundamental marginals, and specified by their respective diagrams, \emph{e.g.,}
\begin{equation}
    \shearedTikZ{\mycustomgrid{3}{3}{v} \node[regular] at (1,1) () {};\node[regular] at (2,1) () {};\node[regular] at (3,1) () {};\node[regular] at (1,2) () {};}, 
    \nbym{3}{2}{v}, \nbym{2}{3}{v}, \ldots
\end{equation}
Let us remark that the presence or absence of the edges are immaterial to our analysis; they are there only to clarify the relative locations of the vertices. 

\subsection{Main result}
\label{sec:results}

We formulate two conditions on the fundamental marginals. First, they are locally consistent:
\begin{equation}
    \nbym{3}{3}{v} \consistent \nbym{3}{3}{v'} 
\end{equation}
$\forall v, v' \in \Lambda.$ The local consistency of these fundamental marginals ensures the local consistency of other fundamental marginals. This is because $\rho_A \consistent \sigma_B$ implies $\rho_{A'} \consistent \sigma_{B'}$ for any $A' \subset A$ and $B'\subset B$. Therefore, for any subset of the $3\times 3$ cluster, its fundamental marginal is uniquely defined.


Next, we assume that the fundamental marginals obey a certain set of quantum Markov chain conditions, denoted as $\mathcal{C}_M$ below. Let us first state the conditions and explain what they mean.
\begin{equation}
\mathcal{C}_{M}^{v} =\left\{
\shearedTikZ{
    \mygrid{v}
    \node[disk] at (0,1) () {};
    \node[square] at (0,0) () {};
    \node[triangle] at (1,0) () {};
    },
\shearedTikZ{
    \mygrid{v}    
    \node[disk] at (0,2) () {};
    \node[disk] at (1,2) () {};
    \node[disk] at (0,1) () {};
    \node[disk] at (0,0) () {};
    \node[square] at (1,0) () {};
    \node[square] at (1,1) () {};
    \node[triangle] at (2,0) () {};
    \node[triangle] at (2,1) () {};
    },
    \shearedTikZ{
    \mygrid{v}    
    \node[disk] at (0,2) () {};
    \node[disk] at (1,2) () {};
    \node[square] at (0,1) () {};
    \node[square] at (1,1) () {};
    \node[triangle] at (2,0) () {};
    \node[triangle] at (2,1) () {};
    \node[triangle] at (0,0) () {};
    \node[triangle] at (1,0) () {};
    },
    \shearedTikZ{
    \mygrid{v}    
    \node[disk] at (2,2) () {};
    \node[square] at (1,2) () {};
    \node[square] at (1,1) () {};
    \node[square] at (2,1) () {};
    \node[triangle] at (0,0) () {};
    \node[triangle] at (1,0) () {};
    \node[triangle] at (2,0) () {};
    \node[triangle] at (0,1) () {};
    \node[triangle] at (0,2) () {};
    }, \,\, + \text{$\pi$-rotations} \right\} \label{eq:markov_conditions}
\end{equation}
\begin{equation}
    \mathcal{C}_M = \bigcup_{v\in \Lambda} \mathcal{C}_{M}^{v}
\end{equation}
Each diagram represents a condition $I(A:C|B)_{\rho}:=S(\rho_{AB}) + S(\rho_{BC}) - S(\rho_B) - S(\rho_{ABC})=0$, where $S(\rho) := -\text{Tr}(\rho \log \rho)$ is the von Neumann entropy and $\rho_{ABC}$ is a fundamental marginal over the set of vertices shown in the diagram, with the following choice of subsystems: red triangle for $A$, green square for $B$, and blue circle for $C$. (Since $I(A:C|B)_{\rho}$ is invariant under the exchange of $A$ and $C$, the blue disks and the red triangles can be exchanged without changing the meaning of the diagram.) The phrase ``$\pi$-rotations'' means the rotated versions of the four diagrams around the center of each underlying cluster by an angle of $\pi$. Let us remark that these conditions are equivalent to the ones advocated in Ref.~\cite{Kim2021}; see Appendix A of the reference for the proof. In the notation of Ref.~\cite{Kim2021} each dot corresponds to a coarse-grained degree of freedom, represented by a square, \emph{e.g.,}
\begin{equation}
      \nbym{3}{3}{} \,\, \text{(this paper)}\,\, =\,\, \centeredTikZ{\squaresempty{3}{3}{0}{0};} \,\,  \left( \text{Ref.~\cite{Kim2021}} \right).
\end{equation}

These conditions have nontrivial implications, which are the main results of this paper. First,  there exists a $\tau \in \mathcal{D}_{\Lambda}$ such that
\begin{equation}
\boxed{
    \tau \consistent \nbym{3}{3}{v}
}
\end{equation}
$\forall v\in \Lambda$. Second, the maximum entropy consistent with the fundamental marginals is
\begin{equation}
\boxed{
    \max_{\substack{\rho \in \mathcal{D}_{\Lambda} \\
    \rho \consistent \mathcal{M}}} S(\rho) = \sum_{v\in \Lambda} \left(\nbym{2}{2}{v}  - \nbym{2}{1}{v} - \nbym{1}{2}{v} + \nbym{1}{1}{v}\right),
}\label{eq:main_entropy_decomposition}
\end{equation}
where $\rho \consistent \mathcal{M}$ means that $\rho$ is consistent with all the fundamental marginals. Thus, from the fundamental marginals, one can obtain a variational upper bound to the ground state energy and more generally, finite-temperature free energy. 

Let us remark that, while the entropy decomposition is formally an infinite sum, for studying finite systems, this sum can be truncated to a finite one by imposing an appropriate boundary condition. For instance, if we want to study a Hamiltonian defined on a finite cluster $C$, we can set the fundamental marginals outside of this cluster as a pure product state, \emph{e.g.,} $|0\rangle\ldots |0\rangle$. This way, the constraints in $\mathcal{C}_M$ defined outside of $C$ are trivially satisfied, and their entropy contribution is precisely zero, making the sum finite. Alternatively, one may choose to impose translational invariance, but with a strictly larger periodicity than the lattice spacing. In that case, the global entropy density will be the summand in Eq.~\eqref{eq:main_entropy_decomposition}, appropriately averaged over the unit cell. 

Thus, our result is a generalization of the existing solutions to the quantum marginal problem (for quantum many-body systems in two spatial dimensions) in two different ways. The solution in Ref.~\cite{Kim2016} is only useful for energy calculation (because no expression for the global entropy was derived), and the one in Ref.~\cite{Kim2021} is only applicable to translationally invariant states. In contrast, our result applies to both energy and free energy calculation and does not require any assumption on translational invariance.

The significance of our result is that $\mathcal{C}_M$ holds universally in a large class of quantum many-body systems that appear in nature~\cite{Kim2021}, at least, empirically and approximately. These are gapped quantum many-body systems in two dimensions, which are challenging to study numerically. Our solution to the quantum marginal problem is likely to apply to these systems (because such systems are expected to obey the entropy scaling law, from which Eq.~\eqref{eq:markov_conditions} follows~\cite{Kim2021}) and allows an efficient calculation of energy and free energy (because both the ground state energy and the global entropy can be calculated from the fundamental marginals directly). This is  evidence that a large class of two-dimensional quantum many-body systems that have remained intractable may be efficiently simulable on a classical computer after all. That would be pleasing, both conceptually and practically.

\section{Conditional independence}
\label{sec:conditional_independence}
Our analysis relies heavily on the concept of \emph{conditional independence}. A tripartite quantum state $\rho_{ABC}$ is conditionally independent (between $A$ and $C$ conditioned on $B$) if $I(A:C|B)_{\rho}=0$. As discussed before, we can alternatively say that $\rho_{ABC}$ is a quantum Markov chain. 

In this Section, we review various facts about conditional independence. While none of them are new, using them is often nontrivial, mostly because of the lack of a convenient formalism. The concept of conditional independence is a genuinely tripartite one, so in a multipartite setup, we must specify the subsystems appropriately. The set of subsets is a large one, and one often has no choice but to specify the subset by its elements. This process requires a large amount of work on properly bookkeeping the subsets. Any notational convention based on an explicit specification of every element in the set quickly becomes unwieldy. 

However, let us recall that the set under our consideration $\Lambda$ is not an arbitrary one. After all, this is a two-dimensional lattice, and a simple diagrammatic notation can convey the same information in a more transparent way. As we briefly discussed in Section~\ref{sec:results}, we shall represent a fundamental marginal in terms of diagrams like this:
\begin{equation}
    \shearedTikZ{\mycustomgrid{3}{3}{v} \node[regular] at (1,1) () {};\node[regular] at (2,1) () {};\node[regular] at (3,1) () {};\node[regular] at (1,2) () {};}, 
    \nbym{3}{2}{v}, \nbym{2}{3}{v}, \ldots
\end{equation}
where the density matrix that each diagram represents is the fundamental marginal over the black dots. The locations of these dots are specified implicitly by the anchoring point $v$. (The edges are merely bookkeeping devices to clarify the relative locations of the vertices.) We can specify the conditional independence relation of the fundamental marginals as, for instance,
\begin{equation}
    \shearedTikZ{
    \mygrid{v}    
    \node[disk] at (2,2) () {};
    \node[square] at (1,2) () {};
    \node[square] at (1,1) () {};
    \node[square] at (2,1) () {};
    \node[triangle] at (0,0) () {};
    \node[triangle] at (1,0) () {};
    \node[triangle] at (2,0) () {};
    \node[triangle] at (0,1) () {};
    \node[triangle] at (0,2) () {};
    },
\end{equation}
which means that $I(A:C|B)_{\rho}=0$ for a fundamental marginal $\rho_{ABC}$ over the $3\times 3$ cluster, where $A$ is the set of red triangles, $B$ is the set of green squares, and $C$ is the set of blue circles. Note that $I(A:C|B)_{\rho} = I(C:A|B)_{\rho}$, so $A$ and $C$ can be exchanged without changing the condition. To write these conditions more explicitly, we would have needed to specify the subsets as $A=\{v, v-\unitx, v-2\unitx, v-2\unitx + \unity, v-2\unitx + 2\unity \},$ $B=\{v+\unity, v+\unity-\unitx, v+2\unity - \unitx\}$, and $C=\{v+2\unity \}$, which is not a particularly appealing convention in the author's opinion.


A fascinating fact is that these diagrams form a basis of an intricate web of logical statements, which we shall explain in detail below. The remainder of this section is organized in the following way. In Section~\ref{sec:mono}, we discuss how SSA can be used to infer new conditional independence relations from the given ones. In Section~\ref{sec:markov_entorpy}, we introduce a notion of Markov entropy decomposition~\cite{Poulin2011} and elucidate its connection to the maximum-entropy principle.~\cite{Kim2014}.  In Section~\ref{sec:quantum_markov_chain} and~\ref{sec:merging_algebra}, we review facts about quantum Markov chains and merging algebra~\cite{Kim2021}.

\subsection{Monotonicity and its reverse}
\label{sec:mono}
From SSA, one can show that
\begin{equation}
\begin{aligned}
    I(A:C|B) &\leq I(A:CD|B), \\
    I(A:C|BD) & \leq I(A:CD|B),
\end{aligned}
\end{equation}
for any density matrix over $\rho_{ABCD}$. (Similarly, one can show that $I(A:C|B) \leq I(AD:C|B)$ and $I(A:C|BD) \leq I(AD:C|B)$.) Because conditional mutual information must be nonnegative by SSA, we can also conclude that
\begin{equation}
\begin{aligned}
     I(A:CD|B)=0 &\to I(A:C|B)=0, \\
    I(A:CD|B)=0 &\to I(A:C|BD)=0. 
\end{aligned}\label{eq:mono_equation}
\end{equation}
We shall refer to the logical deduction of the form in Eq.~\eqref{eq:mono_equation} as the \emph{monotonicity move} and refer to it using $\mono$, \emph{e.g.,}
\begin{equation}
\begin{aligned}
\shearedTikZ{
    \mygrid{v}    
    \node[disk] at (0,2) () {};
    \node[disk] at (1,2) () {};
    \node[disk] at (0,1) () {};
    \node[disk] at (0,0) () {};
    \node[square] at (1,0) () {};
    \node[square] at (1,1) () {};
    \node[triangle] at (2,0) () {};
    \node[triangle] at (2,1) () {};
    }
    \mono
\shearedTikZ{
    \mygrid{v}    
    \node[disk] at (0,2) () {};
    \node[disk] at (1,2) () {};
    \node[disk] at (0,0) () {};
    \node[square] at (1,0) () {};
    \node[square] at (1,1) () {};
    \node[triangle] at (2,0) () {};
    \node[triangle] at (2,1) () {};
    },
\shearedTikZ{
    \mygrid{v}    
    \node[disk] at (0,2) () {};
    \node[disk] at (0,0) () {};
    \node[square] at (1,0) () {};
    \node[square] at (1,1) () {};
    \node[triangle] at (2,0) () {};
    \node[triangle] at (2,1) () {};
    },
\shearedTikZ{
    \mygrid{v}    
    \node[disk] at (0,1) () {};
    \node[disk] at (0,0) () {};
    \node[square] at (1,0) () {};
    \node[square] at (1,1) () {};
    \node[triangle] at (2,0) () {};
    \node[triangle] at (2,1) () {};
    }
    , \ldots\\
    \shearedTikZ{
    \mygrid{v}    
    \node[disk] at (0,2) () {};
    \node[disk] at (1,2) () {};
    \node[disk] at (0,1) () {};
    \node[disk] at (0,0) () {};
    \node[square] at (1,0) () {};
    \node[square] at (1,1) () {};
    \node[triangle] at (2,0) () {};
    \node[triangle] at (2,1) () {};
    } 
    \mono
    \shearedTikZ{
    \mygrid{v}    
    \node[disk] at (0,2) () {};
    \node[square] at (1,2) () {};
    \node[square] at (0,1) () {};
    \node[disk] at (0,0) () {};
    \node[square] at (1,0) () {};
    \node[square] at (1,1) () {};
    \node[triangle] at (2,0) () {};
    \node[triangle] at (2,1) () {};
    } ,
    \shearedTikZ{
    \mygrid{v}    
    \node[disk] at (0,2) () {};
    \node[disk] at (1,2) () {};
    \node[disk] at (0,1) () {};
    \node[disk] at (0,0) () {};
    \node[square] at (1,0) () {};
    \node[square] at (1,1) () {};
    \node[square] at (2,0) () {};
    \node[triangle] at (2,1) () {};
    } 
    ,
    \shearedTikZ{
    \mygrid{v}    
    \node[disk] at (0,2) () {};
    \node[square] at (1,2) () {};
    \node[disk] at (0,1) () {};
    \node[disk] at (0,0) () {};
    \node[square] at (1,0) () {};
    \node[square] at (1,1) () {};
    \node[square] at (2,0) () {};
    \node[triangle] at (2,1) () {};
    } , \ldots
\end{aligned}
\end{equation}
Diagrammatically, we can understand the monotonicity move simply as a process in which (i) parts of the blue circles or red triangles are removed or (ii) they are absorbed into green squares.

Note that the monotonicity argument does not always work in the opposite direction. That is, $I(A:C|B)=0$  generally does not imply $I(A:CD|B)=0$. However, note that $I(A:CD|B) - I(A:C|B) = I(A:D|BC)$. Therefore, $I(A:C|B)=0$ and $I(A:D|BC)=0$ does imply $I(AD:C|B)=0$. (Note that both $I(A:C|B)=0$ and $I(A:D|BC)=0$ are consequences of $I(A:CD|B)=0$.) While neither of the conditions by itself implies $I(A:CD|B)=0$, as a whole they do. We shall refer to this kind of argument as the \emph{reverse monotonicity move} and refer to it as $\revmono$, \emph{e.g.,}
\begin{equation}
    \begin{aligned}
\shearedTikZ{
    \mygrid{v}    
    \node[disk] at (0,2) () {};
    \node[disk] at (1,2) () {};
    \node[disk] at (0,0) () {};
    \node[square] at (1,0) () {};
    \node[square] at (1,1) () {};
    \node[triangle] at (2,0) () {};
    \node[triangle] at (2,1) () {};
    }
    ,
    \shearedTikZ{
    \mygrid{v}    
    \node[square] at (0,2) () {};
    \node[square] at (1,2) () {};
    \node[square] at (0,0) () {};
    \node[square] at (1,0) () {};
    \node[square] at (1,1) () {};
    \node[disk] at (0,1) () {};
    \node[triangle] at (2,0) () {};
    \node[triangle] at (2,1) () {};
    }
    \revmono
    \shearedTikZ{
    \mygrid{v}    
    \node[disk] at (0,2) () {};
    \node[disk] at (1,2) () {};
    \node[disk] at (0,1) () {};
    \node[disk] at (0,0) () {};
    \node[square] at (1,0) () {};
    \node[square] at (1,1) () {};
    \node[triangle] at (2,0) () {};
    \node[triangle] at (2,1) () {};
    },
    \\
    \shearedTikZ{
    \mygrid{v}    
    \node[disk] at (0,2) () {};
    \node[square] at (1,2) () {};
    \node[square] at (0,1) () {};
    \node[disk] at (0,0) () {};
    \node[square] at (1,0) () {};
    \node[square] at (1,1) () {};
    \node[triangle] at (2,0) () {};
    \node[triangle] at (2,1) () {};
    },
    \shearedTikZ{
    \mygrid{v}    
    \node[disk] at (1,2) () {};
    \node[disk] at (0,1) () {};
    \node[square] at (1,0) () {};
    \node[square] at (1,1) () {};
    \node[triangle] at (2,0) () {};
    \node[triangle] at (2,1) () {};
    }
    \revmono
    \shearedTikZ{
    \mygrid{v}    
    \node[disk] at (0,2) () {};
    \node[disk] at (1,2) () {};
    \node[disk] at (0,1) () {};
    \node[disk] at (0,0) () {};
    \node[square] at (1,0) () {};
    \node[square] at (1,1) () {};
    \node[triangle] at (2,0) () {};
    \node[triangle] at (2,1) () {};
    }.
    \end{aligned}
\end{equation}
Diagrammatically, the reverse monotonicity move can be understood in the following way. Two conditions need to be fulfilled. First, the red triangles must be placed on the same set of vertices. Second, the union of the blue circles and green squares of one diagram must be precisely the set of green squares in the other diagram. Once these conditions are met, we get a new diagram in which the green squares are the intersection of the two sets of green squares. The blue circles are the union of the two sets of blue circles.

Note that the reverse monotonicity move can be thought of as an ``opposite'' of the monotonicity move. By reversing the direction of the arrow, we obtain statements that follow from monotonicity.

\subsection{Markov entropy decomposition}
\label{sec:markov_entorpy}
Entropy is a \emph{non-linear} functional of the state. Therefore, it is generally impossible to decompose the entropy into a linear combination of expectation values obtainable from bounded subsystems. However, what one can do in general is to \emph{bound} the entropy in terms of those local quantities. Markov entropy decomposition provides an upper bound of the entropy of a state in terms of its marginal entropies~\cite{Poulin2011}.

Specifically, consider a finite subset of $\Lambda$, denoted as $V$. Note that $V$ induces a set of edges, namely the edges that end on at least one element of $V$. We can start with a set including a single vertex and ``grow'' the set by adding one vertex at a time, among the ones in the neighborhood of that very set. Any such finite sequence defines a \emph{path} $\mathcal{P}_n = (v_1,\ldots, v_n)$, which can be also thought as a sequence of graphs $G_i=(V_i, E_i)$. 

Given a path $\mathcal{P}_n := (v_1, \ldots, v_n),$ the Markov entropy decomposition is defined recursively~\cite{Poulin2011}. Let $\rho \in \mathcal{D}_{V_N}$.
\begin{equation}
    \markoventropy{\rho}{\mathcal{P}_k} = \markoventropy{\rho}{\mathcal{P}_{k-1}} + S(v_k|N(v_k) \cap V_{k-1})_{\rho},
\end{equation}
where $v_k \in N(V_k)$; here $N(v_k)$ is a set of neighbors of $v_k$ and $N(V_k) := \cup_{v\in V_k} N(v_k)$; the Markov entropy decomposition of an empty set is assumed to be $0$. In words, this decomposition iteratively adds the conditional entropy of $v_k$ conditioned on the intersection of its neighbor and the set of vertices that appeared up to that point.

A useful fact about the Markov entropy decomposition is that it is local. Note that the conditional entropy can be computed directly from a neighborhood of a vertex, which contains a bounded number of sites. Therefore, even if the global state is unknown, the Markov entropy decomposition can be computed from the marginals. 

An important fact is that the Markov entropy decomposition upper bounds the von Neumann entropy. This follows from a simple induction argument. The $k=1$ and $2$ case is trivial. For $k\geq 3$ the proof follows from SSA: 
\begin{equation}
\begin{aligned}
    \markoventropy{\rho}{\mathcal{P}_k} - S(\rho_{V_k})&\geq S(\rho_{V_{k-1}}) + S(v_k|N(v_k) \cap V_{k-1})_{\rho} - S(\rho_{V_k}) \\
    &\geq 0.
\end{aligned}
\end{equation}


Another significance of the markov entropy decomposition lies in its connection to the maximum-entropy principle. If the Markov entropy decomposition is equal to the von Neumann entropy of the underlying state, then that state is unique.
Since we use this fact a lot, we state this more formally below, with an outline of the proof.
\begin{lemma}
~\cite{Kim2014}
Let $\mathcal{P}_n= (v_1, \ldots, v_n)$  be a path. Let $\rho, \sigma \in \mathcal{D}_{\{v_1,\ldots, v_n \}}$. If $\rho_{N(v_k)} = \sigma_{N(v_k)}$ for all $v_k \in \mathcal{P}_n$  and $S(\rho) = S(\sigma) = \markoventropy{\rho}{\mathcal{P}_n} = \markoventropy{\sigma}{\mathcal{P}_n}$, $\rho=\sigma$.
\label{lemma:max_entropy_uniqueness}
\end{lemma}
\begin{proof}
Let $\tau = \frac{\rho + \sigma}{2}$. Since $\markoventropy{\tau}{\mathcal{P}_n}$ only depends on the marginal of $\tau$ on $N(v_k)$, $v_k \in \mathcal{P}_n$, $\markoventropy{\tau}{\mathcal{P}_n}= \markoventropy{\rho}{\mathcal{P}_n} = \markoventropy{\sigma}{\mathcal{P}_n}$. Therefore,
\begin{equation}
\begin{aligned}
    S(\tau) - \frac{1}{2} (S(\rho) + S(\sigma)) &\leq \markoventropy{\tau}{\mathcal{P}_n} - \frac{1}{2} (S(\rho) + S(\sigma))\\
    &=0.
\end{aligned}
\end{equation}
Since 
\begin{equation}
    S\left(\frac{\rho + \sigma}{2}\right) - \frac{1}{2} \left(S(\rho) + S(\sigma)\right) \geq \frac{1}{8} \|\rho - \sigma \|_1^2
\end{equation}
for density matrices $\rho$ and $\sigma$, where $\| \cdots \|_1$ is the $1$-norm~\cite{Kim2014}, $\rho=\sigma$.
\end{proof}


\subsection{Quantum Markov Chain}
\label{sec:quantum_markov_chain}

From the conditional independence relation, one can derive nontrivial identities relating different fundamental marginals. Petz showed that $I(A:C|B)_{\rho}=0$ if and only if $\rho_{ABC} = \rho_{BC}^{\frac{1}{2}} \rho_B^{-\frac{1}{2}} \rho_{AB} \rho_B^{-\frac{1}{2}} \rho_{BC}^{\frac{1}{2}}$~\cite{Petz1988}. Here $(\cdot) \to \rho_{BC}^{\frac{1}{2}} \rho_B^{-\frac{1}{2}} (\cdot) \rho_B^{-\frac{1}{2}} \rho_{BC}^{\frac{1}{2}}$ is a quantum channel, known as the \emph{Petz map}. We shall represent this fact as simply 
\begin{equation}
    I(A:C|B)_{\rho} = 0 \Longleftrightarrow \rho_{ABC} = \rho_{AB} \rightmerge \rho_{BC}, \label{eq:petz_theorem}
\end{equation}
where $\sigma_{AB} \rightmerge \rho_{BC} := \rho_{BC}^{\frac{1}{2}} \rho_B^{-\frac{1}{2}} \sigma_{AB} \rho_B^{-\frac{1}{2}} \rho_{BC}^{\frac{1}{2}}$ is the \emph{right-merge} of $\rho_{BC}$ into $\sigma_{AB}$~\cite{Kim2021}. Note that the right-merge operation is well-defined for any pair of density matrices, even if they are not locally consistent. Moreover, Eq.~\eqref{eq:petz_theorem} holds even if we exchange $A$ and $C$ because $I(A:C|B)_{\rho}$ is invariant under such an exchange. We shall represent this logical statement using $\petz$, \emph{e.g.,}
\begin{equation}
    \shearedTikZ{
    \mygrid{v}    
    \node[disk] at (0,2) () {};
    \node[disk] at (1,2) () {};
    \node[disk] at (0,1) () {};
    \node[disk] at (0,0) () {};
    \node[square] at (1,0) () {};
    \node[square] at (1,1) () {};
    \node[triangle] at (2,0) () {};
    \node[triangle] at (2,1) () {};
    } 
    \petz
    \shearedTikZ{
    \mygrid{v}    
    \node[regular] at (0,2) () {};
    \node[regular] at (1,2) () {};
    \node[regular] at (0,1) () {};
    \node[regular] at (0,0) () {};
    \node[regular] at (1,0) () {};
    \node[regular] at (1,1) () {};
    \node[regular] at (2,0) () {};
    \node[regular] at (2,1) () {};
    } 
    =
    \shearedTikZ{
    \mygrid{v}
    \node[regular] at (0,0) () {};
    \node[regular] at (1,0) () {};
    \node[regular] at (0,1) () {};
    \node[regular] at (1,1) () {};
    \node[regular] at (0,2) () {};
    \node[regular] at (1,2) () {};
    }
    \rightmerge
    \shearedTikZ{
    \mygrid{v}
    \node[regular] at (1,0) () {};
    \node[regular] at (2,0) () {};
    \node[regular] at (1,1) () {};
    \node[regular] at (2,1) () {};
    },
\end{equation}
and the same type of equation with an reversed arrow. 

Let us emphasize that the precise form of the Petz map is unimportant. For instance, the rotated Petz map~\cite{Fawzi2015} will serve an equally useful purpose.\footnote{In fact, to extend our argument to the approximate case, \emph{i.e.,} the case in which $I(A:C|B)_{\rho} \approx 0$, it is desirable to use the rotated Petz map instead of the Petz map~\cite{Kim2016}. Working out this modification is left for future work.} Let us state a useful lemma that emphasizes this point.
\begin{lemma}
$I(A:C|B)_{\rho}=0$ if and only if
\begin{equation}
\begin{aligned}
    \text{Tr}_C(\Phi_{B\to BC}(\rho_{AB})) &= \rho_{AB}, \\
    \Phi_{B\to BC}(\rho_{B}) &= \rho_{BC},
\end{aligned}
\end{equation}
for some quantum channel $\Phi_{B\to BC}$.
\label{lemma:local_lemma}
\end{lemma}
\noindent
 We point to Lemma 1 in Ref.~\cite{Kim2021} for the proof of this statement.


\subsection{Merging algebra}
\label{sec:merging_algebra}
A useful fact about the quantum Markov chain is that two Markov chains can be \emph{merged} together into another (larger) Markov chain. Because this new state is again a Markov chain, the merging process can be bootstrapped. This important observation was first made by Kato \emph{et al}~\cite{Kato2016}. In this Section, we will review this fact and its applications~\cite{Kim2017,Kim2021}.

The key is the \emph{merging lemma}, stated below.
\begin{lemma}
(Merging lemma)~\cite{Kato2016} Let $\rho_{ABC} \consistent \sigma_{BCD}$ and $I(A:C|B)_{\rho} = I(B:C|D)_{\sigma}=0$. Then there exists a density matrix $\tau_{ABCD} \consistent \rho_{ABC}, \sigma_{BCD}$ such that $I(A:CD|B)_{\tau} = I(AB:D|C)_{\tau}=0$.
\label{lemma:merging}
\end{lemma}
\noindent
In fact, one can even derive an explicit form of $\tau_{ABCD}$ in terms of $\rho_{ABC}$ and $\sigma_{BCD}$, which is $\tau_{ABCD} = \rho_{ABC} \rightmerge \sigma_{CD} = \sigma_{BCD} \rightmerge \rho_{AB}$. (It is easy to verify, using the monotonicity argument, that $\tau_{ABCD} = \rho_{ABC} \rightmerge \sigma_{BCD} = \sigma_{BCD} \rightmerge \rho_{ABC}$.)


One can iterate the merging lemma by considering a sequence of quantum Markov chains. This leads to a notion of \emph{snake}. Below, Supp($\cdot$) means the support of the density matrix appearing the paranthesis and $\rho \maxmerge \sigma$ is the maximum-entropy state consistent with both $\rho$ and $\sigma$. 
\begin{definition}
\label{def:snake}
Consider a sequence of density matrices $(\rho_i)_{i=1}^N$ such that 
\begin{enumerate}
    \item $\rho_i \maxmerge \rho_{i+1} = \rho_i \rightmerge \rho_{i+1}$ and 
    \item $\text{Supp}(\rho_i) \cap \text{Supp}(\rho_{j})\neq \emptyset$ unless $|i-j|<1$.
\end{enumerate}
Then, we define a snake of $(\rho_i)_{i=1}^N$ as
\begin{equation}
    \mathbb{S}((\rho_i)_{i=1}^N) := ((\rho_1 \rightmerge \rho_2)\cdots ) \rightmerge \rho_N.
\end{equation}
\end{definition}
\noindent
While Definition~\ref{def:snake} is different from the definition of snake in Ref.~\cite{Kim2021}, they are equivalent by the mutation lemma (Lemma 4) of Ref.~\cite{Kim2021}. Our definition will turn out to be more convenient for our proof strategy.

We can be more concrete by considering a specific set of snakes that we can define in our setup. Below, we introduce level-$1$, $2$, and $3$ snakes.
\begin{definition}
\begin{equation}
\text{Snakes: }
\begin{cases}
\text{Level-1: }    \snakeone{v}{u} := \snakeone{v}{u-\unitx} \rightmerge \nbym{2}{1}{u},\\[10pt]
\text{Level-2: }    \snaketwo{v}{u} := \snaketwo{v}{u-\unitx} \rightmerge \nbym{2}{2}{u}, \\[10pt]
\text{Level-3:} \snakethree{v}{u} := \snakethree{v}{u-\unitx} \rightmerge \nbym{2}{3}{u},
\end{cases}
\end{equation}
where 
\begin{equation}
    \snakeone{v}{v+\unitx} = \nbym{2}{1}{v+\unitx}, \,\,\, \snaketwo{v}{v+\unitx} = \nbym{2}{2}{v+\unitx}, \,\,\,
    \snakethree{v}{v+\unitx} = \nbym{2}{3}{v+\unitx}.
\end{equation}
\label{def:snakes_level}
\end{definition}
\noindent
Alternatively, one may define these objects in terms of the \emph{merge product}, which we introduce below.
\begin{equation}
\rho \mergeprod{i=1}{n} \sigma_i := \left( \left( \rho \rightmerge \sigma_1 \right)  \ldots \right) \rightmerge \sigma_n.
\end{equation}
\noindent
In terms of the merge product, the snakes can be rewritten as
\begin{equation}
    \begin{aligned}
    \snakeone{v}{u} &= \nbym{2}{1}{v+\unitx} \mergeprod{i=2}{|u-v|} \nbym{2}{1}{v+i\unitx},\\
    \snaketwo{v}{u} &= \nbym{2}{2}{v+\unitx} \mergeprod{i=2}{|u-v|} \nbym{2}{2}{v+i\unitx},\\
    \snakethree{v}{u} &= \nbym{2}{3}{v+\unitx} \mergeprod{i=2}{|u-v|} \nbym{2}{3}{v+i\unitx}.
    \end{aligned}
    \label{eq:snakes_merge_product}
\end{equation}
Let us also make a brief remark on our convention. Obviously, if $v$ and $u$ have different $y$-coordinates or if $ v_x \geq u_x -1$, these diagrams do not make any sense. Therefore, it must be understood that, whenever these diagrams are used, $v_y=u_y$ and $v_x < u_x-1$. (As a side note, let us remark that we have not yet shown that the level-$1$, $2$, and $3$ snakes are snakes in the sense of Definition~\ref{def:snake}. We defer this proof to Section~\ref{sec:snakes}.)

Now we discuss useful properties of the snakes. First, we note the following entropy decomposition.
\begin{lemma}
\begin{equation}
    \begin{aligned}
    \entropy{\snakeone{v}{u}} &= \markoventropy{\snakeone{v}{u}}{\mathcal{P}}, \\
    \entropy{\snaketwo{v}{u}} &= \markoventropy{\snaketwo{v}{u}}{\mathcal{P}}, \\
    \entropy{\snakethree{v}{u}} &= \markoventropy{\snakethree{v}{u}}{\mathcal{P}}, 
    \end{aligned}
\end{equation}
where $\mathcal{P}$ is a path of columns from the left end to the right end.\footnote{Note that each element in the path contains $1$, $2$, and $3$ vertices for level-$1$, $2$, and $3$ snakes, respectively.}
\label{lemma:entropy_decomposition}
\end{lemma}
\noindent
These decompositions follow from Corollary 1 of Ref.~\cite{Kim2021}, provided that the objects in the parantheses are snakes. That fact will be proven in Section~\ref{sec:snakes}, leading to the proof of Lemma~\ref{lemma:entropy_decomposition}. 

Next, the snakes satisfy the ``splitting property.''\footnote{The point is that a snake can be ``split'' across any of the points in the middle, resulting into two shorter snakes. This process is reversible in the sense that the split snakes can be recombined into the original snake.}
\begin{lemma}
\begin{equation}
\begin{aligned}
    \snakeone{v}{t} &= \snakeone{v}{u} \rightmerge \snakeone{u}{t} = \snakeone{u}{t} \rightmerge \snakeone{v}{u}, \\
    \snaketwo{v}{t} &= \snaketwo{v}{u} \rightmerge \snaketwo{u}{t} = \snaketwo{u}{t} \rightmerge \snaketwo{v}{u}, \\
    \snakethree{v}{t} &= \snakethree{v}{u} \rightmerge \snakethree{u}{t}  = \snakethree{u}{t} \rightmerge \snakethree{v}{u}. 
\end{aligned}
\end{equation}
\label{lemma:splitting}
\end{lemma}
\noindent
These facts follow from Lemma 5 of Ref.~\cite{Kim2021}, provided that one can show that the objects shown are snakes. As stated before, that fact will be proven in Section~\ref{sec:snakes}. 

Repeatedly applying the splitting property, one can show that the order of the merging process can be reversed.
\begin{corollary}
\begin{equation}
\begin{aligned}
\snakeone{v}{u} &= \nbym{2}{1}{u} \mergeprod{i=1}{|u-v|-1} \nbym{2}{1}{u-i\unitx},\\
    \snaketwo{v}{u} &= \nbym{2}{2}{u} \mergeprod{i=1}{|u-v|-1} \nbym{2}{2}{u-i\unitx},\\
    \snakethree{v}{u} &= \nbym{2}{3}{u} \mergeprod{i=1}{|u-v|-1} \nbym{2}{3}{u-i\unitx}.
\end{aligned}
\end{equation}
\label{corollary:reversal}
\end{corollary}

\subsection{Are they snakes?}
\label{sec:snakes}

By ``they'' we mean the level-$1, 2$, and $3$ snakes. We will show this in the affirmative, justifying their namesake. Let $v\in \Lambda$. Here is a three-line proof for the level-$1$ snake.
\begin{equation}
\shearedTikZ{
    \mygrid{v}    
    \node[disk] at (0,2) () {};
    \node[disk] at (1,2) () {};
    \node[disk] at (0,1) () {};
    \node[disk] at (0,0) () {};
    \node[square] at (1,0) () {};
    \node[square] at (1,1) () {};
    \node[triangle] at (2,0) () {};
    \node[triangle] at (2,1) () {};
    } 
    \mono  
    \shearedTikZ{
    \mygrid{v}    
    \node[disk] at (0,1) () {};
    \node[square] at (1,0) () {};
    \node[square] at (1,1) () {};
    \node[triangle] at (2,0) () {};
    \node[triangle] at (2,1) () {};
    }.
\end{equation}
\begin{equation}
   \shearedTikZ{
    \mygrid{v}    
    \node[disk] at (0,1) () {};
    \node[square] at (1,0) () {};
    \node[square] at (1,1) () {};
    \node[triangle] at (2,0) () {};
    \node[triangle] at (2,1) () {};
    }, 
    \shearedTikZ{
    \mygrid{v}    
    \node[disk] at (0,1) () {};
    \node[triangle] at (1,0) () {};
    \node[square] at (1,1) () {};
    }
    \revmono
    \shearedTikZ{
    \mygrid{v}    
    \node[disk] at (0,1) () {};
    \node[triangle] at (1,0) () {};
    \node[square] at (1,1) () {};
    \node[triangle] at (2,0) () {};
    \node[triangle] at (2,1) () {};
    }. \label{eq:rev1}
\end{equation}
\begin{equation}
    \shearedTikZ{
    \mygrid{v}    
    \node[disk] at (0,1) () {};
    \node[triangle] at (1,0) () {};
    \node[square] at (1,1) () {};
    \node[triangle] at (2,0) () {};
    \node[triangle] at (2,1) () {};
    }
    \mono 
    \shearedTikZ{
    \mygrid{v}    
    \node[disk] at (0,1) () {};
    \node[square] at (1,1) () {};
    \node[triangle] at (2,1) () {}; \label{eq:littlesnake1}
    }.
\end{equation}
Therefore, the level-$1$ snake is indeed a snake. 

The proof for the level-$2$ snake follows from a one-line argument:
\begin{equation}
    \shearedTikZ{
    \mygrid{v}    
    \node[disk] at (0,2) () {};
    \node[disk] at (1,2) () {};
    \node[disk] at (0,1) () {};
    \node[disk] at (0,0) () {};
    \node[square] at (1,0) () {};
    \node[square] at (1,1) () {};
    \node[triangle] at (2,0) () {};
    \node[triangle] at (2,1) () {};
    } 
    \mono  
    \shearedTikZ{
    \mygrid{v}    
    \node[disk] at (0,1) () {};
    \node[disk] at (0,0) () {};
    \node[square] at (1,0) () {};
    \node[square] at (1,1) () {};
    \node[triangle] at (2,0) () {};
    \node[triangle] at (2,1) () {};
    }.
\end{equation}

For the level-$3$ snake, note the following identity:
\begin{equation}
    \nbym{3}{3}{v} = 
    \left( \shearedTikZ{
    \mycustomgrid{3}{3}{v}
    \node[regular] at (1,1) () {};
    \node[regular] at (1,2) () {};
    \node[regular] at (1,3) () {};
    \node[regular] at (2,1) () {};
    \node[regular] at (2,2) () {};
    \node[regular] at (2,3) () {};
    } \rightmerge 
    \shearedTikZ{
    \mycustomgrid{3}{3}{v}
    \node[regular] at (2,1) () {};
    \node[regular] at (2,2) () {};
    \node[regular] at (3,1) () {};
    \node[regular] at (3,2) () {};
    }
    \right)
    \rightmerge 
    \shearedTikZ{
    \mycustomgrid{3}{3}{v}
    \node[regular] at (2,2) () {};
    \node[regular] at (2,3) () {};
    \node[regular] at (3,2) () {};
    \node[regular] at (3,3) () {};
    },
\end{equation}
which follows straightforwardly from the Markovian constraints imposed on the fundamental marginals. Viewing the two right-merges as a quantum channel, using Lemma~\ref{lemma:local_lemma}, we conclude
\begin{equation}
    \shearedTikZ{
    \mycustomgrid{3}{3}{v}
    \node[disk] at (1,1) () {};
    \node[disk] at (1,2) () {};
    \node[disk] at (1,3) () {};
    \node[square] at (2,1) () {};
    \node[square] at (2,2) () {};
    \node[square] at (2,3) () {};
    \node[triangle] at (3,1) () {};
    \node[triangle] at (3,2) () {};
    \node[triangle] at (3,3) () {};
    }.
\end{equation}
Thus, the level-$3$ snake is indeed a snake.

\section{Recursion relations}
\label{sec:recursion}
Armed with the tools in Section~\ref{sec:conditional_independence}, we can now build up a global state consistent with the fundamental marginals. The results in this Section are generalizations of Proposition 1, 2, and 3 in Ref.~\cite{Kim2021}, which will yield a generalization of the main result of Ref.~\cite{Kim2021}, as we explain in Section~\ref{sec:final}.

\subsection{Level $1$ $\to$ Level $2$}
\label{sec:one_to_two}

In this Section, we establish a method to obtain a level-$2$ snake from a level-$1$ snake. Specifically, define the following objects, which we shall refer to as the \emph{flat diagrams.} (This is to be contrasted with other objects that will be soon defined, which are referred to as the \emph{hooked diagrams}.)
\begin{definition}
\begin{equation}
    \begin{aligned}
    \left[\snaketwo{v}{u} \right]_{\uparrow}&:= \snakeone{v}{u} \left[ \mergeprod{i=1}{|u-v|} \nbym{2}{2}{v+i\unitx} \right],
   \\
   \left[\snaketwo{v}{u} \right]_{\downarrow}&:= \snakeone{v+\unity}{u+\unity} \left[\mergeprod{i=0}{|u-v|-1} \nbym{2}{2}{u-i\unitx} \right].
    \end{aligned}
\end{equation}
\end{definition}
\noindent
We show that
\begin{restatable}[]{proposition}{propone}
\label{prop:one_to_two}
\begin{equation}
\begin{aligned}
   \snaketwo{v}{u}&= \left[ \snaketwo{v}{u} \right]_{\uparrow} = \left[ \snaketwo{v}{u} \right]_{\downarrow}.
\end{aligned}
\label{eq:leveltwo_main_statement}
\end{equation}
\end{restatable}

To prove these claims, it will be easier to work with a \emph{hooked diagram} from $v$ to $u$, defined as 
\begin{equation}
    \left[\snaketwo{v}{u}\right]_{\uparrow, \text{hook}} := \snakeonehooked{v}{u} \left[ \mergeprod{i=1}{|u-v|} \nbym{2}{2}{v+i\unitx} \right],
\end{equation}
where
\begin{equation}
    \snakeonehooked{v}{u} := \snakeonehooked{v}{u-\unitx} \rightmerge
    \shearedTikZ{\mycustomgrid{2}{2}{u}; \node[regular] at (1,1) () {}; \node[regular] at (2,1) () {};}
\end{equation}
The $\pi$-rotated version of the hooked diagram is defined in an analogous way, replacing $\uparrow$ by $\downarrow$. Below, we focus on proving the first identity. The second identity follows from the same logic, by rotating all the diagrams involved in the ensuing argument by $\pi$.

Let us first remark that the hooked diagram is equal to the flat diagram.
\begin{lemma}
\begin{equation}
    \left[\snaketwo{v}{u}\right]_{\uparrow} = \left[\snaketwo{v}{u}\right]_{\uparrow, \text{hook}}. \label{eq:leveltwo_hook_nohook}
\end{equation}
\label{lemma:leveltwo_hook_nohook}
\end{lemma}
\begin{proof}
To derive this identity, it will be convenient to rearrange the two merging sequences associated with the left- and the right-hand-side of Eq.~\eqref{eq:leveltwo_hook_nohook}, by commuting through the right-merges of the $2\times 2$ fundamental marginals all the way to the left (until it becomes impossible to commute with anything that lies to its left). Then the two merging sequences differ only over the first three fundamental marginals, \emph{i.e.,}
\begin{equation}
    \begin{aligned}
    \left(\shearedTikZ{
    \mycustomgrid{3}{2}{v+2\unitx}
    \node[regular] at (1,1) () {};
    \node[regular] at (2,1) () {};
    }
    \rightmerge
    \shearedTikZ{
    \mycustomgrid{3}{2}{v+2\unitx}
    \node[regular] at (2,1) () {};
    \node[regular] at (3,1) () {};
    }
    \right) \rightmerge 
    \shearedTikZ{
    \mycustomgrid{3}{2}{v+2\unitx}
    \node[regular] at (1,1) () {};
    \node[regular] at (2,1) () {};
    \node[regular] at (1,2) () {};
    \node[regular] at (2,2) () {};
    } \,\,\, & \text{ for } \left[ \snaketwo{u}{v} \right]_{\uparrow} \,\, \text{and} \\
    \left(\shearedTikZ{
    \mycustomgrid{3}{2}{v+2\unitx}
    \node[regular] at (1,1) () {};
    \node[regular] at (2,1) () {};
    \node[regular] at (1,2) () {};
    }
    \rightmerge
    \shearedTikZ{
    \mycustomgrid{3}{2}{v+2\unitx}
    \node[regular] at (2,1) () {};
    \node[regular] at (3,1) () {};
    }
    \right) \rightmerge 
    \shearedTikZ{
    \mycustomgrid{3}{2}{v+2\unitx}
    \node[regular] at (1,1) () {};
    \node[regular] at (2,1) () {};
    \node[regular] at (1,2) () {};
    \node[regular] at (2,2) () {};
    } \,\,\, & \text{ for }  \left[ \snaketwo{u}{v} \right]_{\uparrow, \text{hook}}.
    \end{aligned}
    \label{eq:merging_hook_unhook_comparison_local}
\end{equation}

The two terms in Eq.~\eqref{eq:merging_hook_unhook_comparison_local} can be shown to be both equal to the fundamental marginal over $\shearedTikZmini{\mycustomgrid{3}{2}{} \node[regular] at (1,1) () {}; \node[regular] at (2,1) () {}; \node[regular] at (3,1) () {}; \node[regular] at (1,2) () {}; \node[regular] at (2,2) () {};}$ anchored at $v+2\unitx$. For the first term, this fact follows from
\begin{equation}
    \shearedTikZ{\mycustomgrid{3}{2}{v+2\unitx}
    \node[disk] at (1,1) () {};
    \node[square] at (2,1) () {};
    \node[triangle] at (3,1) () {};
    }
    \,\,\,
    \text{and }
    \,\,\,
    \shearedTikZ{\mycustomgrid{3}{2}{v+2\unitx}
    \node[disk] at (1,2) () {};
    \node[disk] at (2,2) () {};
    \node[square] at (1,1) () {};
    \node[square] at (2,1) () {};
    \node[triangle] at (3,1) () {};
    },
\end{equation}
each following from Eq.~\eqref{eq:littlesnake1} and the $\pi$-rotated version of Eq.~\eqref{eq:rev1}, \emph{i.e.,}
\begin{equation}
    \shearedTikZ{\mycustomgrid{3}{2}{v+2\unitx}
    \node[disk] at (1,2) () {};
    \node[disk] at (2,2) () {};
    \node[disk] at (1,1) () {};
    \node[square] at (2,1) () {};
    \node[triangle] at (3,1) () {};
    }
    \mono
    \shearedTikZ{\mycustomgrid{3}{2}{v+2\unitx}
    \node[disk] at (1,2) () {};
    \node[disk] at (2,2) () {};
    \node[square] at (1,1) () {};
    \node[square] at (2,1) () {};
    \node[triangle] at (3,1) () {};
    },
\end{equation}
respectively. The second term in Eq.~\eqref{eq:merging_hook_unhook_comparison_local} follows from 
\begin{equation}
   \shearedTikZ{\mycustomgrid{3}{2}{v+2\unitx}
    \node[disk] at (1,2) () {};
    \node[disk] at (2,2) () {};
    \node[disk] at (1,1) () {};
    \node[square] at (2,1) () {};
    \node[triangle] at (3,1) () {};
    }
    \mono
    \shearedTikZ{\mycustomgrid{3}{2}{v+2\unitx}
    \node[disk] at (1,2) () {};
    \node[disk] at (1,1) () {};
    \node[square] at (2,1) () {};
    \node[triangle] at (3,1) () {};
    }, 
    \shearedTikZ{\mycustomgrid{3}{2}{v+2\unitx}
    \node[square] at (1,2) () {};
    \node[disk] at (2,2) () {};
    \node[square] at (1,1) () {};
    \node[square] at (2,1) () {};
    \node[triangle] at (3,1) () {};}.
    \label{eq:eq01}
\end{equation}
Thus, Eq.~\eqref{eq:leveltwo_hook_nohook} is proven. 
\end{proof}

Now we are in a position to prove Proposition~\ref{prop:one_to_two}. By Lemma~\ref{lemma:leveltwo_hook_nohook}, it suffices to prove
\begin{equation}
    \snaketwo{v}{u} =  \left[\snaketwo{v}{u} \right]_{\uparrow, \text{hook}}.\label{eq:level_two_intermediate_goal}
\end{equation}
Here is a high-level overview of the proof strategy behind Eq.~\eqref{eq:level_two_intermediate_goal}. The basic idea is to use the maximum-entropy principle. Specifically, we shall show that the hooked level-$2$ snake and the level-$2$ snake have the same marginals over every $2\times 2$ clusters strictly contained in their supports; moreover they have the same global entropy, which is equal to the Markov entropy decomposition. By Lemma~\ref{lemma:max_entropy_uniqueness}, the two states are the same. 

There are two useful facts.
\begin{lemma}
\begin{equation}
\text{Tr}_{v, v+\unity} \left( \left[\snaketwo{v}{u} \right]_{\uparrow, \text{hook}}\right) = \left[\snaketwo{v+\unitx}{u} \right]_{\uparrow, \text{hook}}
\end{equation}
\label{lemma:recursion_level_two}
\end{lemma}
\begin{proof}
After using the rearrangement of the right-merges discussed in the proof of Lemma~\ref{lemma:leveltwo_hook_nohook}, Lemma~\ref{lemma:recursion_level_two} follows immediately from 
\begin{equation}
    \shearedTikZ{
    \mycustomgrid{3}{2}{v+2\unitx}
    \node[regular] at (1,1) () {};
    \node[regular] at (2,1) () {};
    \node[regular] at (3,1) () {};
    \node[regular] at (1,2) () {};
    \node[regular] at (2,2) () {};
    } = 
    \left(
    \shearedTikZ{
    \mycustomgrid{3}{2}{v+2\unitx}
    \node[regular] at (1,1) () {};
    \node[regular] at (2,1) () {};
    \node[regular] at (1,2) () {};
    } \rightmerge 
    \shearedTikZ{
    \mycustomgrid{3}{2}{v+2\unitx}
    \node[regular] at (2,1) () {};
    \node[regular] at (3,1) () {};
    }
    \right)
    \rightmerge
    \shearedTikZ{
    \mycustomgrid{3}{2}{v+2\unitx}
    \node[regular] at (1,1) () {};
    \node[regular] at (2,1) () {};
    \node[regular] at (1,2) () {};
    \node[regular] at (2,2) () {};
    },
\end{equation}
a fact already proved in Eq.~\eqref{eq:eq01}.
\end{proof}
\begin{lemma}
\begin{equation}
\text{Tr}_{t: t_x>v_x+1} \left( \left[\snaketwo{v}{u} \right]_{\uparrow, \text{hook}}\right) = \nbym{2}{2}{v+\unitx}
\end{equation}
\label{lemma:recursion_level_two_local}
\end{lemma}
\begin{proof}
Note that the partial trace of the hooked level-$2$ snake is equal to the reduced density matrix of $\left(\shearedTikZmini{\mycustomgrid{4}{2}{} \node[regular] at (1,1) () {}; \node[regular] at (2,1) () {}; \node[regular] at (3,1) () {}; \node[regular] at (1,2) () {}; \node[regular] at (2,2) () {};} \rightmerge \shearedTikZmini{\mycustomgrid{4}{2}{} \node[regular] at (3,1) () {}; \node[regular] at (4,1) () {};}\right) \rightmerge \shearedTikZmini{\mycustomgrid{4}{2}{}  \node[regular] at (2,1) () {}; \node[regular] at (3,1) () {}; \node[regular] at (2,2) () {}; \node[regular] at (3,2) () {};}$ over the left-most $2\times 2$ cluster. (The bottom-right corners of these diagrams are all $v+3\unitx$.) Tracing out the right-most vertex yields $\shearedTikZmini{\mycustomgrid{4}{2}{} \node[regular] at (1,1) () {}; \node[regular] at (2,1) () {}; \node[regular] at (3,1) () {}; \node[regular] at (1,2) () {}; \node[regular] at (2,2) () {};}$ because one can merge (using Lemma~\ref{lemma:merging}) $\shearedTikZmini{\mycustomgrid{4}{2}{} \node[regular] at (1,1) () {}; \node[regular] at (2,1) () {}; \node[regular] at (3,1) () {}; \node[regular] at (1,2) () {}; \node[regular] at (2,2) () {};}$ and $\shearedTikZmini{\mycustomgrid{4}{2}{}  \node[regular] at (2,1) () {}; \node[regular] at (3,1) () {};  \node[regular] at (2,2) () {}; \node[regular] at (4,1) () {};}$; note that the two marginals are locally consistent and satisfy the following conditional independence relations:
\begin{equation}
    \begin{aligned}
    \shearedTikZ{
    \mygrid{}    
    \node[disk] at (0,2) () {};
    \node[disk] at (1,2) () {};
    \node[disk] at (0,1) () {};
    \node[disk] at (0,0) () {};
    \node[square] at (1,0) () {};
    \node[square] at (1,1) () {};
    \node[triangle] at (2,0) () {};
    \node[triangle] at (2,1) () {};
    } &\mono
    \shearedTikZ{
    \mygrid{}    
    \node[disk] at (0,1) () {};
    \node[disk] at (0,0) () {};
    \node[square] at (1,0) () {};
    \node[square] at (1,1) () {};
    \node[triangle] at (2,0) () {};
    }, \\
    \shearedTikZ{
    \mygrid{}
    \node[square] at (1,1) () {};
    \node[disk] at (2,1) () {};
    \node[triangle] at (0,1) () {};
    \node[triangle] at (0,2)  () {};
    \node[triangle] at (1,2) () {};
    }
    &\mono
    \shearedTikZ{
    \mygrid{}
    \node[square] at (1,1) () {};
    \node[triangle] at (0,1) () {};
    \node[triangle] at (0,2)  () {};
    \node[disk] at (2,1) () {};
    }.
    \end{aligned}
\end{equation}
(The second line follows from an argument identical to the derivation of Eq.~\eqref{eq:rev1}, up to a rotation by $\pi$.) Thus, the reduced density matrix of $\left(\shearedTikZmini{\mycustomgrid{4}{2}{} \node[regular] at (1,1) () {}; \node[regular] at (2,1) () {}; \node[regular] at (3,1) () {}; \node[regular] at (1,2) () {}; \node[regular] at (2,2) () {};} \rightmerge \shearedTikZmini{\mycustomgrid{4}{2}{} \node[regular] at (3,1) () {}; \node[regular] at (4,1) () {};}\right) \rightmerge \shearedTikZmini{\mycustomgrid{4}{2}{}  \node[regular] at (2,1) () {}; \node[regular] at (3,1) () {}; \node[regular] at (2,2) () {}; \node[regular] at (3,2) () {};}$ over the left-most $2\times 2$ cluster is equal to that of $\shearedTikZmini{\mycustomgrid{4}{2}{} \node[regular] at (1,1) () {}; \node[regular] at (2,1) () {}; \node[regular] at (3,1) () {}; \node[regular] at (1,2) () {}; \node[regular] at (2,2) () {};} \rightmerge \shearedTikZmini{\mycustomgrid{4}{2}{}  \node[regular] at (2,1) () {}; \node[regular] at (3,1) () {}; \node[regular] at (2,2) () {}; \node[regular] at (3,2) () {};}$. Using a conditional independence relation that follows from monotonicity, one can show that $\shearedTikZmini{\mycustomgrid{4}{2}{} \node[regular] at (1,1) () {}; \node[regular] at (2,1) () {}; \node[regular] at (3,1) () {}; \node[regular] at (1,2) () {}; \node[regular] at (2,2) () {};} \rightmerge \shearedTikZmini{\mycustomgrid{4}{2}{}  \node[regular] at (2,1) () {}; \node[regular] at (3,1) () {}; \node[regular] at (2,2) () {}; \node[regular] at (3,2) () {};}$ is equal to $\shearedTikZmini{\mycustomgrid{4}{2}{} \node[regular] at (1,1) () {}; \node[regular] at (2,1) () {}; \node[regular] at (3,1) () {}; \node[regular] at (1,2) () {}; \node[regular] at (2,2) () {}; \node[regular] at (3,2) () {};}$. Thus, the reduced density matrix of a level $2$-snake is consistent with the fundamental marginal on the left-most $2\times 2$ cluster. This completes the proof of Lemma~\ref{lemma:recursion_level_two_local}.
\end{proof}

Now we are in a position to prove Proposition~\ref{prop:one_to_two}. Let us restate the proposition for the reader's convenience.
\propone*
\begin{proof}
Note that 
\begin{equation}
 \left[ \snaketwo{v}{u} \right]_{\uparrow, \text{hook}}    =  \Phi_{v,u}^{(2)}\left(\nbym{2}{2}{v+\unitx}\right)
\end{equation}
for some quantum channel $\Phi^{(2)}_{v, u}$ from the second column to itself and the columns that extend to the $u$ and a vertex above it. Moreover, taking a partial trace over all but the first two columns, we again obtain $\shearedTikZmini{\mycustomgrid{2}{2}{}; \node[regular] at (1,1) () {}; \node[regular] at (2,1) () {}; \node[regular] at (1,2) () {}; \node[regular] at (2,2) () {};}$ anchored at $v+\unitx$; see Lemma~\ref{lemma:recursion_level_two_local}. Thus, we can apply Lemma~\ref{lemma:local_lemma} to conclude that the hooked diagram forms a Markov chain. (Specifically, we can conclude that $I(A:C|B)_{\rho}=0$ for the hooked-snake $\rho$, with $A$ being the first column, $B$ being the second column, and $C$ being the rest.) Using the entropy decomposition in Lemma~\ref{lemma:entropy_decomposition}, we get
\begin{equation}
    \entropy{\left[ \snaketwo{v}{u} \right]_{\uparrow, \text{hook}}} = \entropy{\left[ \snaketwo{v+\unitx}{u} \right]_{\uparrow, \text{hook}}} + \entropy{\nbym{2}{2}{v+\unitx}} - \entropy{\nbym{1}{2}{v+\unitx}}.
\end{equation}
Iterating this argument, we find that the entropy of the hooked diagram is equal to its Markov entropy decomposition whose underlying path is a sequence of columns from $v_x$ to $u_x$. Recall that the entropy of the level-$2$ snake is identical to this Markov entropy decomposition; see Lemma~\ref{lemma:entropy_decomposition}. Since both states are consistent on the marginals that appear in the Markov entropy decomposition and both are maximum-entropy states, by Lemma~\ref{lemma:max_entropy_uniqueness}
\begin{equation}
    \snaketwo{u}{v} = \left[ \snaketwo{u}{v} \right]_{\uparrow, \text{hook}}.
\end{equation}
By the equivalence of the hooked diagram and the flat diagram, we conclude that
\begin{equation}
    \snaketwo{u}{v} = \left[ \snaketwo{u}{v} \right]_{\uparrow}.
\end{equation}
The proof of the $\pi$-rotated version follows the exact same argument.
\end{proof}

\subsection{Level-$2$ $\to$ Level-$1$}
\label{sec:two_to_one}
In this Section, we prove the ``reverse'' of Proposition~\ref{prop:one_to_two}.
\begin{restatable}[]{proposition}{proptwo}
\begin{equation}
\begin{aligned}
    \snakeone{v}{u} &= \text{Tr}_{\substack{\{t\}: t_y=v_y+1}}\left( \snaketwo{v}{u} \right) \\
    &= \text{Tr}_{\substack{\{t\}: t_y=v_y-1}}\left( \snaketwo{v-\unity}{u-\unity}. \right)
\end{aligned}
    \label{eq:two_to_one}
\end{equation}
\label{prop:two_to_one}
\end{restatable}
\begin{proof}
Below, we prove the equivalence of the level-$1$ snake to the second line of Eq.~\eqref{eq:two_to_one}. The equivalence of the level-$1$ snake to the first line can be proved in an analogous way. The very first right-merge for the level-$2$ snake results in $\shearedTikZmini{\mycustomgrid{3}{2}{}; \node[regular] at (1,1) () {}; \node[regular] at (2,1) () {}; \node[regular] at (3,1) () {}; \node[regular] at (1,2) () {}; \node[regular] at (2,2) () {}; \node[regular] at (3,2) () {};}$ anchored at $v-\unity + \unitx$. Upon tracing out the bottom-left corner, we obtain
\begin{equation}
    \shearedTikZ{
    \mycustomgrid{3}{2}{t}
    \node[regular] at (1,2) () {};
    \node[regular] at (2,2) () {};
    \node[regular] at (3,2) () {};
    \node[regular] at (3,1) () {};
    \node[regular] at (2,1) () {};
    } = 
    \shearedTikZ{
    \mycustomgrid{3}{2}{t}
    \node[regular] at (2,1) () {};
    \node[regular] at (2,2) () {};
    \node[regular] at (3,1) () {};
    \node[regular] at (3,2) () {};
    }
    \rightmerge
    \shearedTikZ{
    \mycustomgrid{3}{2}{t}
    \node[regular] at (1,2) () {};
    \node[regular] at (2,2) () {};
    },
\end{equation}
for $t= v-\unity + \unity$, wherein the equality follows from Eq.~\eqref{eq:rev1}. Note that the support of $\shearedTikZmini{
    \mycustomgrid{3}{2}{}
    \node[regular] at (1,2) () {};
    \node[regular] at (2,2) () {};
    }$ (again anchored at $v-\unity + \unity$) is disjoint from the support of $\shearedTikZmini{\mycustomgrid{2}{2}{}; \node[regular] at (1,1) () {}; \node[regular] at (2,1) () {}; \node[regular] at (1,2) () {}; \node[regular] at (2,2) () {};}$ associated with the next merging operation; therefore, these two can be exchanged. Repeating this argument, we can push the right-merge of $\shearedTikZmini{
    \mycustomgrid{3}{2}{}
    \node[regular] at (1,2) () {};
    \node[regular] at (2,2) () {};
    }$ all the way to the end, obtaining
\begin{equation}
    \text{Tr}_{v}\left( \snaketwo{v}{u} \right) = \snaketwo{v+\unitx}{u} \rightmerge 
    \shearedTikZ{
    \mycustomgrid{2}{2}{v+\unitx}
    \node[regular] at (1, 2) () {};
    \node[regular] at (2, 2) () {};
    }.
\end{equation}
Repeating the same argument, we obtain a ``reversed'' version of level-$1$ snake. By Corollary~\ref{corollary:reversal}, this is equal to the level-$1$ snake, proving Proposition~\ref{prop:two_to_one}.
\end{proof}

\subsection{Level-$2$ $\to$ Level-$3$}
\label{sec:two_to_three}
In Section~\ref{sec:one_to_two}, we showed that level-$2$ snakes can be obtained by applying a sequence of right-merge operations over $2\times 2$ clusters on the level-$1$ snake. We prove an analogous statement for the level-$3$ snakes in this Section. The main statement and the proof strategy share many parallels with those in Section~\ref{sec:one_to_two}. We state the key definition and statement below. Again, define the \emph{flat-diagrams}. 
\begin{definition}
\begin{equation}
\begin{aligned}
\left[ \snakethree{v}{u} \right]_{\uparrow}&:= \snaketwo{v}{u} \left[ \mergeprod{i=1}{|u-v|} \nbym{2}{2}{v+i\unitx + \unity} \right] \\
\left[ \snakethree{v}{u} \right]_{\downarrow}&:= \snaketwo{v+\unity}{u+\unity} \left[\mergeprod{i=0}{|u-v|-1} \nbym{2}{2}{u-i\unitx} \right]
\end{aligned}
\end{equation}
\end{definition}

The main result of this Section is the following proposition.
\begin{restatable}[]{proposition}{propthree}
\label{prop:two_to_three}
\begin{equation}
    \begin{aligned}
    \snakethree{v}{u} &=  \left[ \snakethree{v}{u} \right]_{\uparrow}\\
    &= \left[ \snakethree{v}{u} \right]_{\downarrow}.
    \end{aligned}
\end{equation}
\end{restatable}
\noindent
As in Section~\ref{sec:one_to_two}, we focus on proving the first equality in Proposition~\ref{prop:two_to_three}. The second equality follows the same logic, by rotating the diagrams involved in the argument by $\pi$. 

It will be convenient to consider, again, the \emph{hooked diagrams}:
\begin{equation}
   \snakethree{v}{u} := \snaketwohooked{v}{u} \left[ \mergeprod{i=1}{|u-v|} \nbym{2}{2}{v+i\unitx + \unity} \right],
\end{equation}
where
\begin{equation}
    \snaketwohooked{v}{u} = \snaketwohooked{v}{u-\unitx} \rightmerge
    \shearedTikZ{\mycustomgrid{2}{3}{u};
    \node[regular] at (1,1) () {};
    \node[regular] at (2,1) () {};
    \node[regular] at (1,2) () {};
    \node[regular] at (2,2) () {};
    }.
\end{equation}
Let us remark that the hooked diagram is equal to the flat diagram.
\begin{lemma}
\begin{equation}
    \left[\snakethree{v}{u}\right]_{\uparrow} = \left[\snakethree{v}{u}\right]_{\uparrow, \text{hook}}. \label{eq:levelthree_hook_nohook}
\end{equation}
\label{lemma:levelthree_hook_nohook}
\end{lemma}
\begin{proof}
Let us rearrange the two merging sequences associated with the left- and the right-hand-side of Eq.~\eqref{eq:levelthree_hook_nohook}, by commuting through the right-merges of the $2\times 2$ fundamental marginals all the way to the left. Then the two merging sequences differ only over the first three fundamental marginals, \emph{i.e.,}
\begin{equation}
    \begin{aligned}
    \left( \shearedTikZ{
    \mycustomgrid{3}{3}{v+2\unitx};
    \node[regular] at (1,1) () {};
    \node[regular] at (2,1) () {};
    \node[regular] at (1,2) () {};
    \node[regular] at (2,2) () {};
    } 
    \rightmerge
    \shearedTikZ{
    \mycustomgrid{3}{3}{v+2\unitx};
    \node[regular] at (2,1) () {};
    \node[regular] at (2,2) () {};
    \node[regular] at (3,1) () {};
    \node[regular] at (3,2) () {};
    }
    \right)
    \rightmerge
    \shearedTikZ{
    \mycustomgrid{3}{3}{v+2\unitx};
    \node[regular] at (1,2) () {};
    \node[regular] at (2,2) () {};
    \node[regular] at (1,3) () {};
    \node[regular] at (2,3) () {};
    } \,\,\, &\text{ for }
    \left[ \snakethree{v}{u} \right]_{\uparrow},
    \\
    \left( \shearedTikZ{
    \mycustomgrid{3}{3}{v+2\unitx};
    \node[regular] at (1,1) () {};
    \node[regular] at (2,1) () {};
    \node[regular] at (1,2) () {};
    \node[regular] at (2,2) () {};
    \node[regular] at (1,3) () {};
    } 
    \rightmerge
    \shearedTikZ{
    \mycustomgrid{3}{3}{v+2\unitx};
    \node[regular] at (2,1) () {};
    \node[regular] at (2,2) () {};
    \node[regular] at (3,1) () {};
    \node[regular] at (3,2) () {};
    }
    \right)
    \rightmerge
    \shearedTikZ{
    \mycustomgrid{3}{3}{v+2\unitx};
    \node[regular] at (1,2) () {};
    \node[regular] at (2,2) () {};
    \node[regular] at (1,3) () {};
    \node[regular] at (2,3) () {};
    }
    \,\,\, &\text{ for } 
    \left[ \snakethree{v}{u} \right]_{\uparrow, \text{hook}}.
    \end{aligned}
    \label{eq:823}
\end{equation}
Both of them are equal to the fundamental marginal $\shearedTikZmini{
\mycustomgrid{3}{3}{};
\node[regular] at (1,1) () {};
\node[regular] at (2,1) () {};
\node[regular] at (3,1) () {};
\node[regular] at (1,2) () {};
\node[regular] at (2,2) () {};
\node[regular] at (3,2) () {};
\node[regular] at (1,3) () {};
\node[regular] at (2,3) () {};
}$ anchored at $v+2\unitx$. To see why, note the following two conditional independence conditions for establishing this fact for the first case:
\begin{equation}
    \begin{aligned}
    &\shearedTikZ{\mycustomgrid{3}{3}{v+2\unitx};
    \node[disk] at (1,1) () {};
    \node[disk] at (1,2) () {};
    \node[disk] at (1,3) () {};
    \node[disk] at (2,3) () {};
    \node[square] at (2,1) () {};
    \node[square] at (2,2) () {};
    \node[triangle] at (3,1) () {};
    \node[triangle] at (3,2) () {};
    }
    \mono
    \shearedTikZ{\mycustomgrid{3}{3}{v+2\unitx};
    \node[disk] at (1,1) () {};
    \node[disk] at (1,2) () {};
    \node[square] at (2,1) () {};
    \node[square] at (2,2) () {};
    \node[triangle] at (3,1) () {};
    \node[triangle] at (3,2) () {};
    } \,\, \text{ and } \\
    &\shearedTikZ{\mycustomgrid{3}{3}{v+2\unitx};
    \node[disk] at (1,3) () {};
    \node[disk] at (2,3) () {};
    \node[square] at (1,2) () {};
    \node[square] at (2,2) () {};
    \node[triangle] at (1,1) () {};
    \node[triangle] at (2,1) () {};
    \node[triangle] at (3,1) () {};
    \node[triangle] at (3,2) () {};
    },
    \end{aligned}
\end{equation}
from which the equivalence (between the fundamental marginal and the first line of Eq.~\eqref{eq:823}) follows. For the second case, note that
\begin{equation}
    \begin{aligned}
    &\shearedTikZ{\mycustomgrid{3}{3}{v+2\unitx};
    \node[disk] at (1,1) () {};
    \node[disk] at (1,2) () {};
    \node[disk] at (1,3) () {};
    \node[disk] at (2,3) () {};
    \node[square] at (2,1) () {};
    \node[square] at (2,2) () {};
    \node[triangle] at (3,1) () {};
    \node[triangle] at (3,2) () {};
    }
    \mono
    \shearedTikZ{\mycustomgrid{3}{3}{v+2\unitx};
    \node[disk] at (1,1) () {};
    \node[disk] at (1,2) () {};
    \node[disk] at (1,3) () {};
    \node[square] at (2,1) () {};
    \node[square] at (2,2) () {};
    \node[triangle] at (3,1) () {};
    \node[triangle] at (3,2) () {};
    } \,\, \text{ and } \\
    &\shearedTikZ{\mycustomgrid{3}{3}{v+2\unitx};
    \node[disk] at (1,3) () {};
    \node[disk] at (2,3) () {};
    \node[square] at (1,2) () {};
    \node[square] at (2,2) () {};
    \node[triangle] at (1,1) () {};
    \node[triangle] at (2,1) () {};
    \node[triangle] at (3,1) () {};
    \node[triangle] at (3,2) () {};
    }
    \mono
    \shearedTikZ{\mycustomgrid{3}{3}{v+2\unitx};
    \node[square] at (1,3) () {};
    \node[disk] at (2,3) () {};
    \node[square] at (1,2) () {};
    \node[square] at (2,2) () {};
    \node[triangle] at (1,1) () {};
    \node[triangle] at (2,1) () {};
    \node[triangle] at (3,1) () {};
    \node[triangle] at (3,2) () {};
    }.
    \end{aligned}
    \label{eq:849}
\end{equation}
Thus, the main claim is proved.
\end{proof}

Now, similar to what we did in Section~\ref{sec:one_to_two}, we will prove a certain recursion relation. Specifically, we wish to use the following two facts: (i) tracing out the first column yields a shorter hooked diagram and (ii) the marginal of the hooked diagram over the first two columns is equal to the fundamental marginal defined on that cluster.

\begin{lemma}
\begin{equation}
\text{Tr}_{v, v+\unity, v+2\unity}\left( \left[ \snakethree{v}{u} \right]_{\uparrow, \text{hook}} \right)
= \left[ \snakethree{v+\unitx}{u} \right]_{\uparrow, \text{hook}}
\end{equation}
\label{lemma:snake_to_snake_three}
\end{lemma}
\begin{proof}
Rearrange the right-merges by commuting through the right-merges of the $2\times 2$ marginals with an anchoring point having $y$-coordinate of $v_y+1$, all the way to the left. After this rearrangement, the first two right-merges yields the fundamental marginal $\shearedTikZmini{
\mycustomgrid{3}{3}{};
\node[regular] at (1,1) () {};
\node[regular] at (2,1) () {};
\node[regular] at (3,1) () {};
\node[regular] at (1,2) () {};
\node[regular] at (2,2) () {};
\node[regular] at (3,2) () {};
\node[regular] at (1,3) () {};
\node[regular] at (2,3) () {};
}$ anchored at $v+2\unitx$; see Eq.~\eqref{eq:849}. Upon tracing out the first column, we obtain $\shearedTikZmini{
\mycustomgrid{2}{3}{};
\node[regular] at (1,1) () {};
\node[regular] at (2,1) () {};
\node[regular] at (1,2) () {};
\node[regular] at (2,2) () {};
\node[regular] at (1,3) () {};
}$ anchored at $v+2\unitx$, completing the proof.
\end{proof}

\begin{lemma}
\begin{equation}
    \text{Tr}_{t: t_x>v_x+1} \left( \left[\snakethree{v}{u} \right] \right) = \nbym{2}{3}{v+\unitx}
\end{equation}
\label{lemma:snakethree_local}
\end{lemma}
\begin{proof}
Again use the rearrangement discussed in the proof of Lemma~\ref{lemma:snake_to_snake_three}. The result of this partial trace is equal to the reduced density matrix of $\left(\shearedTikZmini{\mycustomgrid{4}{3}{} \node[regular] at (1,1) () {}; \node[regular] at (2,1) () {}; \node[regular] at (3,1) () {}; \node[regular] at (1,2) () {}; \node[regular] at (2,2) () {}; \node[regular] at (3,2) () {}; \node[regular] at (1,3) () {}; \node[regular] at (2,3) () {};} \rightmerge \shearedTikZmini{\mycustomgrid{4}{3}{}  \node[regular] at (3,1) () {};  \node[regular] at (3,2) () {}; \node[regular] at (4,1) () {}; \node[regular] at (4,2) {};}\right) \rightmerge \shearedTikZmini{\mycustomgrid{4}{3}{}  \node[regular] at (2,2) () {};  \node[regular] at (2,3) () {}; \node[regular] at (3,2) () {}; \node[regular] at (3,3) {};}$ over the left-most $2\times 3$ cluster. (The bottom-right corner of these diagrams are $v+3\unitx$.)  Note that the partial trace of the fourth column effectively removes the first right-merge. This is because we can use the following conditional independence relations to merge $\shearedTikZmini{\mycustomgrid{4}{3}{} \node[regular] at (1,1) () {}; \node[regular] at (2,1) () {}; \node[regular] at (3,1) () {}; \node[regular] at (1,2) () {}; \node[regular] at (2,2) () {}; \node[regular] at (3,2) () {}; \node[regular] at (1,3) () {}; \node[regular] at (2,3) () {};}$ and $\shearedTikZmini{\mycustomgrid{4}{3}{}  \node[regular] at (2,1) () {}; \node[regular] at (3,1) () {};  \node[regular] at (2,2) () {}; \node[regular] at (3,2) () {}; \node[regular] at (4,1) () {}; \node[regular] at (4,2) {};}$ (The bottom-right corners are again $v+3\unitx$.):
\begin{equation}
    \shearedTikZ{\mycustomgrid{4}{3}{v+3\unitx} 
    \node[disk] at (1,1) () {}; 
    \node[square] at (2,1) () {}; 
    \node[triangle] at (3,1) () {}; 
    \node[disk] at (1,2) () {}; 
    \node[square] at (2,2) () {}; 
    \node[triangle] at (3,2) () {}; 
    \node[disk] at (1,3) () {}; 
    \node[disk] at (2,3) () {};}
    \,\, \text{and} \,\, 
    \shearedTikZ{\mycustomgrid{4}{3}{v+3\unitx}  
    \node[disk] at (2,1) () {}; 
    \node[square] at (3,1) () {};  
    \node[disk] at (2,2) () {}; 
    \node[square] at (3,2) () {}; 
    \node[triangle] at (4,1) () {}; 
    \node[triangle] at (4,2) {};}.
\end{equation}
After this partial trace, we obtain
\begin{equation}
\shearedTikZ{
\mycustomgrid{3}{3}{v+2\unitx};
\node[triangle] at (1,1) () {};
\node[triangle] at (2,1) () {};
\node[triangle] at (3,1) () {};
\node[triangle] at (1,2) () {};
\node[triangle] at (1,3) () {};
\node[square] at (2,2) () {};
\node[square] at (2,3) () {};
\node[square] at (3,2) () {};
\node[disk] at (3,3) () {};
}
\petz
    \shearedTikZ{
    \mycustomgrid{3}{3}{v+2\unitx} \node[regular] at (1,1) () {}; \node[regular] at (2,1) () {}; \node[regular] at (3,1) () {}; \node[regular] at (1,2) () {}; \node[regular] at (2,2) () {}; \node[regular] at (3,2) () {}; \node[regular] at (1,3) () {}; \node[regular] at (2,3) () {};
    }
    \rightmerge 
    \shearedTikZ{
    \mycustomgrid{3}{3}{v+2\unitx}  \node[regular] at (2,2) () {};  \node[regular] at (2,3) () {}; \node[regular] at (3,2) () {}; \node[regular] at (3,3) {};
    } = \nbym{3}{3}{v+2\unitx}.
\end{equation}
By definition, the reduced density matrix over the left-most $2\times 3$ cluster is the fundamental marginal over that cluster.
\end{proof}

Now we are in a position to prove Proposition~\ref{prop:two_to_three}.
\propthree*

\begin{proof}
Since the flat diagram is equivalent to the corresponding hooked diagram (Lemma~\ref{lemma:levelthree_hook_nohook}), it suffices to prove the equivalence of level-$3$ snakes to their respective hooked diagrams. This is what we prove below.
Note that
\begin{equation}
    \left[ \snakethree{v}{u} \right]_{\uparrow, \text{hook}} = \Phi_{v,u}^{(3)}\left( \nbym{2}{3}{v+\unitx} \right)
\end{equation}
for some quantum channel $\Phi_{v,u}^{(3)}$ from the the second column of the diagram to itself and the columns that extend all the way up to $u$. Moreover, this channel has the property that (i) upon tracing out all the vertices from the third to the last column, one obtains the fundamental marginal over a $2\times 3$ cluster anchored at $v+\unitx$ (Lemma~\ref{lemma:snakethree_local}) and (ii) upon tracing out the first column, one obtains a shorter hooked diagram (Lemma~\ref{lemma:snake_to_snake_three}). Thus, we can apply Lemma~\ref{lemma:local_lemma} to conclude that the hooked diagram forms a Markov chain. Using the entropy decomposition for the Markov chain iteratively, we find that the entropy of the hooked diagram is equal to its Markov entropy decomposition. (The underlying path is a sequence of columns from $v_x$ to $u_x$.) The same entropy decomposition applies to the the level-$3$ snake; see Lemma~\ref{lemma:entropy_decomposition}. Since both states are consistent on the marginals that appear in the Markov entropy decomposition and both are maximum-entropy states, by Lemma~\ref{lemma:max_entropy_uniqueness}, the main claim follows.
\end{proof}

\section{Main proofs}
\label{sec:final}
In this Section, we combine the results discussed in Section~\ref{sec:recursion} to obtain our main results. Here are the three key statements.
\propone*
\proptwo*
\propthree*

With these three statements together, we can show that the level-$3$ snake in the vertical direction forms a quantum Markov chain.
\begin{lemma}
\label{lemma:snake_vertical}
\begin{equation}
    \snakethreecmi{u}{v}
\end{equation}
\end{lemma}
\begin{proof}
\begin{equation}
\begin{aligned}
\text{Tr}_{t: t_y = v_y+2} \left( \left[\snakethree{u}{v} \right]_{\uparrow} \right) 
&= \text{Tr}_{t: t_y = v_y+2} \left( \left[\snakethree{u}{v} \right]_{\downarrow} \right) \\
&= \left[\snaketwo{u}{v} \right]_{\downarrow} \\
&= \snaketwo{u}{v}.
\end{aligned}
\end{equation}
From Lemma~\ref{lemma:local_lemma}, the claim follows immediately.
\end{proof}

\begin{theorem}
\label{thm:main1}
The fundamental marginals are consistent with some global state.
\end{theorem}
\begin{proof}
The level-$3$ snakes are consistent with the fundamental marginals. By Lemma~\ref{lemma:snake_vertical}, the following density matrix is consistent with all the level-$3$ snakes:
\begin{equation}
    \snaketwo{u}{v} \left[ \mergeprod{i=1}{\infty} \snaketwo{u+i\unity}{v+i\unity} \right].
\end{equation}
for any $u$ and $v$ (such that they have the same $y$-coordinates and $v_x > u_x$). Taking $u = (-\infty, -\infty)$ and $v= (\infty, -\infty)$, the proof follows.
\end{proof}

\begin{theorem}
\label{thm:main2}
Let $\rho$ be the maximum-entropy state consistent with the fundamental marginals. Then 
\begin{equation}
     S(\rho) = \sum_{v\in \Lambda} \left(\nbym{2}{2}{v}  - \nbym{2}{1}{v} - \nbym{1}{2}{v} + \nbym{1}{1}{v}\right).
\end{equation}
\end{theorem}
\begin{proof}
This entropy is achievable. From Lemma~\ref{lemma:snake_vertical} and Lemma~\ref{lemma:entropy_decomposition}, we can decompose the maximum global entropy into the entropies of level-$2$ snakes, which can be again broken down to the entropies of the fundamental marginals, via Lemma~\ref{lemma:entropy_decomposition}.

At the same time, this entropy is the maximum entropy achievable, which follows from the Markov entropy decomposition using a path from $u=(-\infty, u_y)$ to $v= (\infty, u_y)$ for $u_y=-\infty$ and then concatenating the same path for $u_y+1, u_y+2,$ etc.
\end{proof}
\noindent
Note that the limit of the infinite merge product and infinite sum must be taken carefully. One concrete way to think about these proofs is to set $u=(-N, -M)$ and $v=(N,-M)$, let $i$ range from $1$ to $M$, and then take the $N, M\to \infty$ limit.

\section{Discussion}
\label{sec:discussion}
We generalized the solution in Ref.~\cite{Kim2021} by replacing the translational invariance condition by the local consistency condition. With this improvement, the only conditions needed to ensure the consistency of marginals are (i) their local consistency and (ii) the Markovian constraints $\mathcal{C}_M$, inherited from the entropy scaling law. Thus, symmetry plays no role in our solution; only entropy does.

\section*{Acknowledgement}
I thank Daniel Ranard for the discussion. I also thank Jiri Guth Jarkovsky for spotting a typo.

\bibliographystyle{myhamsplain2}
\bibliography{bib}
\end{document}